\documentclass{article}[12pt]
\usepackage[utf8]{inputenc}
\usepackage{amsmath, amssymb, amscd, amsthm, amsfonts}  
\usepackage{ upgreek }
\usepackage{graphicx} 
\usepackage{subcaption} 
\usepackage{hyperref} 
\usepackage{algorithm} 
\usepackage{algorithmicx}
\usepackage{algpseudocode} 
\usepackage{caption} 
\usepackage{float} 
\usepackage{indentfirst} 
\usepackage[T1]{fontenc} 
\usepackage[utf8]{inputenc} 
\usepackage{framed}
\usepackage{authblk} 
\newcommand\keywords[1]{\textbf{Keywords}: #1}

\newtheorem{theorem}{Theorem}
\newtheorem{lemma}{Lemma}
\newtheorem{definition}{Definition}

\title{\textbf{The Algorithm for Solving Quantum Linear Systems of Equations With Coherent Superposition and Its Extended Applications}}
\author[1,2,3]{Qiqing Xia}
\author[1,2,3]{Qianru Zhu}
\author[4]{Huiqin Xie}
\author[1,2]{Li Yang  \thanks{Corresponding author: yangli@iie.ac.cn}}
\affil[1]{Key Laboratory of Cyberspace Security Defense, Beijing, China}
\affil[2]{Institute of Information Engineering, Chinese Academy of Sciences, Beijing, China}
\affil[3]{School of Cyber Security, University of Chinese Academy of Sciences, Beijing, China}
\affil[4]{ Beijing Electronic Science and Technology Institute, Beijing, China}
\date{ }

\begin{document}

\maketitle

\smallskip
\hrule
\smallskip

\begin{abstract}
    Many quantum algorithms for attacking symmetric cryptography involve the rank problem of quantum linear equations. In this paper, we first propose two quantum algorithms for solving quantum linear systems of equations with coherent superposition and construct their specific quantum circuits. Unlike previous related works, our quantum algorithms are universal. Specifically, the two quantum algorithms can both compute the rank and general solution by one measurement. The difference between them is whether the data register containing the quantum coefficient matrix can be disentangled with other registers and keep the data qubits unchanged. On this basis, we apply the two quantum algorithms as a subroutine to parallel Simon's algorithm (with multiple periods), Grover Meets Simon algorithm, and Alg-PolyQ2 algorithm, respectively. Afterwards, we construct a quantum classifier within Grover Meets Simon algorithm and the test oracle within Alg-PolyQ2 algorithm in detail, including their respective quantum circuits. To our knowledge, no such specific analysis has been done before. We rigorously analyze the success probability of those algorithms to ensure that the success probability based on the proposed quantum algorithms will not be lower than that of those original algorithms. Finally, we discuss the lower bound of the number of CNOT gates for solving quantum linear systems of equations with coherent superposition, and our quantum algorithms reach the optimum in terms of minimizing the number of CNOT gates. Furthermore, our analysis indicates that the proposed algorithms are mainly suitable for conducting attacks against lightweight symmetric ciphers, within the effective working time of an ion trap quantum computer.
\end{abstract}


\section{Introduction}
Quantum computing fully utilizes the advantages of parallel computing by exploiting quantum phenomena such as quantum superposition and entanglement,
which allows for simultaneous processing of multiple computable states, enhancing efficiency and speed. The unique capabilities of quantum computers allow them to outperform classical computers in some computational tasks. They can solve some problems that are difficult for traditional computers in polynomial time, which significantly impacts the security of many cryptographic schemes nowadays. The well-known quantum algorithm proposed by Shor \cite{shor1994algorithms} allows solving factorization and discrete logarithm problems in polynomial time and can achieve an exponential acceleration effect. Since the security of almost all public key schemes in the current classical cryptosystem relies on this computational assumption that those problems are intractable, Shor's algorithm seriously threatens the security of public key cryptography. The quantum algorithm proposed by Grover \cite{grover1996fast,grover1997quantum} can achieve quadratic acceleration compared to the classical exhaustive search. People thought that this was the only threat to symmetric cryptography. However, with many cryptanalyses relying on Simon's algorithm \cite{kuwakado2010quantum, kaplan2016breaking, bonnetain2018quantum, bonnetain2019quantum}, people changed their minds. Nowadays, the impact of many quantum attacks in symmetric cryptography is still not so clear, and possible potential threats continue to be explored.

The quantum algorithm with exponential speedup proposed by Simon \cite{simon1997power} is the enlightenment of many quantum algorithms, and it is of great significance in symmetric cryptanalysis. However, \cite{beals2001quantum} pointed out that only a problem with a promise can obtain the exponential speedup in advance, and any quantum algorithm with unlimited Boolean functions can only provide higher polynomial speedup over classical deterministic algorithms. Simon's algorithm is often used in cryptanalysis of symmetric cryptography and attacks certain constructions of cipher modes often involved in symmetric cryptography. Kuwakado and Morii first used it to break the 3-round Feistel scheme \cite{kuwakado2010quantum} and then proved that the Even-Mansour construction \cite{kuwakado2012security} was insecure with superposition query. Santoli and Schaffiner extended their results and proposed a quantum forgery attack on the CBC-MAC scheme \cite{santoli2017using}. In \cite{kaplan2016breaking}, Kaplan et al. used Simon's algorithm to attack various symmetric cryptosystems, such as CBC-MAC, PMAC, GMAC, GCM, OCB, etc. At the same time, it was applied to many constructions, such as LRW Construction. Simon's algorithm was also applied in the sliding attack \cite{bonnetain2020quantum} to achieve an exponential acceleration effect. Leander et al. \cite{leander2017grover} combined the quantum algorithm of Grover and Simon to attack FX-Construction for the first time in a cryptographic setting, thereby destroying this construction with whitened keys. In \cite{bonnetain2019quantum}, Bonnetain et al. proposed the Alg-PolyQ2 algorithm to attack FX-Construction, greatly reducing the query complexity compared to Grover Meets Simon algorithm.

The problem of solving linear equations must be involved in the process of using Simon's algorithm. Solving linear systems of equations is the central issue of scientific computing, and the solutions of linear equations in the quantum era need further research. In the quantum setting, Harrow et al. \cite{harrow2009quantum} proposed a quantum algorithm to solve non-homogeneous linear systems of equations with a sparse Hamiltonian coefficient matrix, reducing the time complexity of the classical algorithm from $\mathcal{O}(n) $ to $\mathcal{O}(\log(n))$. Based on the idea of HHL algorithm and quantum simulation algorithm \cite{berry2014exponential}, Childs et al. \cite{childs2017quantum} proposed a new algorithm to avoid the limitation of HHL algorithm phase estimation and exponentially improve the dependence on accuracy. Then, Wossnig et al. \cite{wossnig2018quantum} proposed a new algorithm using the quantum singular value estimation algorithm (QSVE) in \cite{DBLP:conf/innovations/KerenidisP17} to break through the sparse matrix assumption in HHL algorithm and further improve the algorithm for solving quantum linear systems of equations. Subasi et al. \cite{subacsi2019quantum} proposed two quantum algorithms based on adiabatic quantum computing to solve linear equations, further reducing the time complexity.

In the field of cryptography, we often need the accurate solution of linear equations, and solving linear equations is a subroutine in many cryptanalytic algorithms. For cryptanalysis in public-key cryptosystems, the intermediate process of some classical information set decoding schemes often needs to use the Gaussian elimination algorithm to transform the matrix. Perriello \cite{perriello2021complete} and Esser \cite{DBLP:journals/corr/abs-2112-06157} respectively constructed quantum circuits for the information set decoding problem in code-based cryptanalysis algorithms and gave an implementation for solving a system of full-rank quantum linear equations with coherent superposition in the corresponding process. Regarding the cryptanalysis of symmetric cryptosystems, many quantum algorithms for attacking symmetric ciphers involve the problem of quantum linear equations, especially in the process of using Simon's algorithm. Simon's algorithm serves as a subroutine in various quantum cryptanalysis algorithms. It needs to solve the problem of quantum linear equations in the quantum setting. In \cite{bonnetain2020quantum}, it proposed the method of judging whether the quantum linear equations is full rank, which is of great help to many algorithms based on the quantum linear systems of equations. In \cite{bonnetain2019quantum}, it gave two algorithms for asymmetric search of a shift using Simon's algorithm under the Q1 and Q2 models, which involves the rank problem of quantum linear equations. Many cryptographic protocols are also designed based on quantum linear equations, such as quantum key distribution \cite{bennett2014quantum}, quantum authentication \cite{curty2001quantum} and quantum secret sharing \cite{wang2005quantum}, etc. It can be seen that solving quantum linear systems of equations is a significant research problem in cryptography.

\textbf{Our contributions} In this paper, we define quantum linear equations with coherent superposition and propose two quantum algorithms for solving quantum linear equations with coherent superposition in the quantum setting, which are developed from the algorithm in \cite{bonnetain2021quantum}. Then, we modify and extend this original algorithm to obtain different versions and apply them to different quantum symmetric attack algorithms as a subroutine. In more detail, our main contributions are as follows:

\begin{enumerate}
    \item We propose two quantum algorithms, differing in whether the data register containing the quantum coefficient matrix can be disentangled with other registers and keep the data qubits unchanged. Algorithm \ref{Alg:QAFGSAR} is the case where the data register is entangled with other registers, but it can reduce quantum resources. Algorithm \ref{Alg:QAFGSARAV} requires that the data register is disentangle with other registers (the data qubits unchanged before and after performing) and adds about $\mathcal{O}(n^2)$ qubits to store the general solution. Different from the previous quantum algorithms for solving linear equations, the algorithms we propose are universal for solving quantum linear equations with coherent superposition in the quantum setting, which is the same as the effect of using Gaussian elimination to solve any classical linear equations in a classical computer. Our method is similar to the one given by Bonnetain et al. in \cite{bonnetain2021quantum} to solve the quantum circuit of Boolean linear equations, but we can compute the rank and general solution in any case under the condition of unchanged complexity algorithm in the quantum setting, and give a detailed circuit construction. After performing $\mathcal{O}(n^3)$ quantum gates, we can obtain the rank, general solution, and upper triangular matrix by one measurement, and the register containing the solutions can also be used in the subsequent circuits of other quantum algorithms, which makes our method more general.
    
    \item We give three applications involving quantum linear equations as a subroutine, respectively parallel Simon's algorithm \cite{simon1997power} (with multiple periods), Grover Meets Simon algorithm \cite{leander2017grover}, and Alg-PolyQ2 algorithm \cite{bonnetain2019quantum}. First, applying Algorithm \ref{Alg:QAFGSAR} to parallel Simon's algorithm (with multiple periods) can solve the problem by one measurement with the original success probability unchanged (or even higher, close to 1). Similarly, We reconstruct the classifier (in Grover Meets Simon algorithm) as a new quantum classifier by applying Algorithm \ref{Alg:QAFGSAR}, including its detailed quantum circuit. Afterwards, we construct a detailed circuit of the \textbf{test} oracle (in Alg-PolyQ2 algorithm) combined with Algorithm \ref{Alg:QAFGSARAV} (which can compute the rank, even though the original article does not specifically introduce how to compute the rank). Moreover, we also construct the specific quantum circuits of the three applications. The solutions to these problems can be obtained by one measurement using the proposed algorithms. Finally, we discuss the optimality of our algorithms and their applicability to lightweight cryptographic attacks during the effective working time of an ion trap quantum computer.
\end{enumerate}

\textbf{Outline} Section \ref{Preliminaries} introduces some knowledge of linear algebra and the basic principles of quantum computing and quantum circuits. Section \ref{Algorithms} gives the specific definition of quantum linear equations with coherent superposition and proposes two quantum algorithms, including their circuit constructions. Section \ref{Application} applies the proposed quantum algorithms in some quantum symmetric attack cryptographic algorithms and constructs their specific quantum circuits. Section \ref{Discussion and Conclusion} discusses the lower bound of quantum gate resources for quantum linear equations and the applications of some lightweight ciphers in an ion trap quantum computer, then summarizes the whole article.

\section{Preliminaries} \label{Preliminaries}
In this section, we briefly recall some notations and results about linear equations over binary fields and quantum circuits.

\subsection{Gaussian Elimination over Binary Fields}
Among the classical algorithms for solving linear equations, the most common one is the Gaussian elimination algorithm. Here, we introduce Gaussian elimination algorithm over binary fields.

\begin{definition}
\textbf{(Linear Equations over Binary Fields)} Let $A$ be a $m \times n$ matrix, and its elements $a_{i,j} \in \{0,1\}$, then $A$ is called a matrix over binary fields. Given a matrix $A$ and a constant vector $b$, find a vector $x$ such that $Ax=b$. When $b=\textbf{0}$, it is called a homogeneous linear system of equations; when $b \neq \textbf{0}$, it is called a non-homogeneous linear system of equations.
\end{definition}

\begin{definition}
\textbf{(Row Echelon Form over Binary Fields)} Given a $m \times n$ matrix $A$ over binary fields, $A$ is called row echelon form over binary fields if the following conditions are satisfied:
\begin{enumerate}
    \item The first non-zero element of each row is 1, and its column index strictly increases with the increase of the row index (the column index must not be smaller than the row index).
    \item All elements below the first non-zero element of the non-zero row are zero.
    \item Any rows consisting of all zeros appear below the rows with non-zero elements.
\end{enumerate}
It can be written in the following form:
\begin{center}
$\begin{bmatrix} a_{1,1} & a_{1,2}& \cdots &a_{1,r}&a_{1,r+1}& \cdots a_{1,n}
\\0 & a_{2,2} & \cdots &a_{2,r}&a_{2,r+1}& \cdots a_{2,n}\\
\vdots & \vdots& \ddots&\vdots&\vdots&\vdots\\
0 & 0& \cdots &a_{r,r} & a_{r,r+1} & a_{r,n}\\
0 & 0& \cdots & 0 & 0 & 0\\
\vdots & \vdots& \ddots&\vdots&\vdots&\vdots\\
0 & 0  & \cdots & 0 &0&0
\end{bmatrix}_{m\times n}$ 
\end{center}
\end{definition}

\begin{definition}
\textbf{(Row Reduced Form Matrix)} Given a row echelon matrix over binary fields, if the first non-zero elements of its non-zero rows are all 1, and the rest elements in the column where the first element 1 of its non-zero row is located are all 0, the matrix is called a row reduced form matrix.
\end{definition}

\begin{theorem}\label{theorem1}
If the rank of the coefficient matrix $A_{m \times n}$ is equal to the rank of the coefficient augmented matrix $[A|b]_{m \times (n+1)}$ over binary fields, i.e., $rank(A) =rank([A|b])$, then this linear system of equations must have a set of solutions.
\end{theorem}
\begin{proof}
Let $rank(A)$ is the rank of $A$, $rank([A|b])$ is the rank of $[A|b]$, $n$ is the number of unknown variables, then $rank(A)\leq n$, $rank([A|b])\leq n+1$.

If $rank(A)=rank([A|b])$, then the last column in the augmented matrix $[A|b]$ must be a linear combination of the first $n$ columns, i.e., there is a set of solutions to the linear system of equations, and the theorem holds.

If $rank(A)<rank([A|b])$, i.e., $rank(A)+1=rank([A|b])$, then the rank of $[A|b]$ is equal to the rank of $A$ add the dimension of the subspace spanned by the column vector $b$ in $A$, so $b$ is not in the subspace spanned by the column vectors of $A$, i.e., $b$ cannot be linearly represented by the column vectors of $A$, then the linear system of equations has no solution.

Therefore, when $rank_A = rank_{[A|b]}$, the linear system of equations must have a set of solutions.
\end{proof}

\begin{theorem}\label{theorem2}
If there are $k$ basic solution vectors of the homogeneous linear system of equations $Ax=0$ over binary fields, $x$ has a zero solution and $2^k-1$ non-zero solutions. Then, x has $2^k$ non-zero solutions for non-homogeneous linear system of equations $Ax=b$.
\end{theorem}

\subsection{Quantum Circuit}
We briefly introduce the relevant knowledge of quantum circuits. A circuit composed of multiple quantum gates with certain logic functions is called a quantum circuit. It is applied to a series of operations of a group of qubits and can be used to describe the change of the quantum state in the two-dimensional Hilbert space. Each quantum logic gate can be represented by a unitary matrix.

A single qubit has two quantum ground states $|0\rangle$,$|1\rangle$. If the quantum state $|\varphi\rangle$ is in a state other than the ground state and can be represented linearly by $|0\rangle$ and $|1\rangle$, then the state is called a superposition state $|\varphi\rangle= \alpha|0\rangle+\beta|1\rangle$. The probability amplitudes $\alpha$ and $\beta$ are complex numbers and satisfy $|\alpha|^2+|\beta|^2=1$.

When given $n$ qubits, the computational basis has $2^n$ vectors. After applying a series of quantum gates to the initial state, the original quantum superposition state is changed. Finally, we measure the quantum system and get some $n$-qubit vectors on the computational ground state to obtain our desired results.

In quantum circuits, the unitary operators in the Hilbert space are all reversible, and it may be necessary to use auxiliary qubits to achieve our desired goal. Typically, we can perform some computations, copy the results to an output register using Controlled-NOT (CNOT) gates, and uncompute (performing the same operation backwards) to restore the initial state of the auxiliary qubits. That is, if there is a computational operation $U$, then the uncomputing operation is a hermitian conjugate transpose operator $U^\dagger$. Since arbitrary quantum gates can be represented by single-qubit quantum gates and CNOT gates, we focus on these quantum gates. We first show single-qubit quantum gates in Figure \ref{fig:SQG}.
\begin{figure}[htbp]
    \centering
    \includegraphics[scale=0.3]{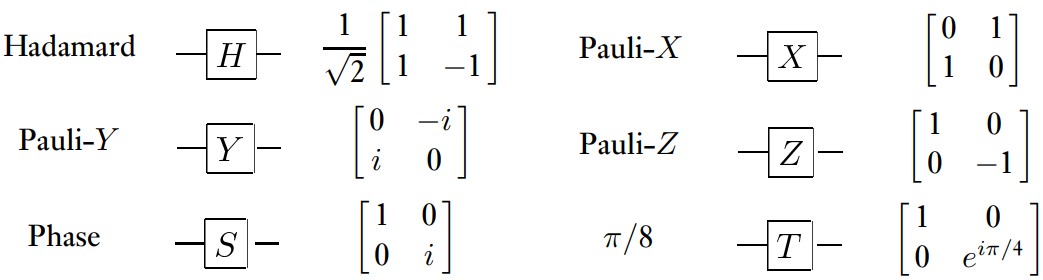}
    \caption{Single-qubit quantum gates}
    \label{fig:SQG}
\end{figure}

In quantum circuits, multi-qubit gates are often used. Common multi-qubit gates include CNOT gates, Toffoli gates, SWAP gates, and Fredkin gates. The XOR gate in the classical circuit can be implemented by the CNOT gate in the quantum logic gates. Similarly, the Toffoli gate can implement the "AND" operation in quantum computing, and it can also be regarded as a double-controlled CNOT gate. The Fredkin gate can be viewed as a controlled SWAP gate. They are shown in Figure \ref{fig:MQG}.
\begin{figure}[htbp]
\centerline{\includegraphics[scale=0.35]{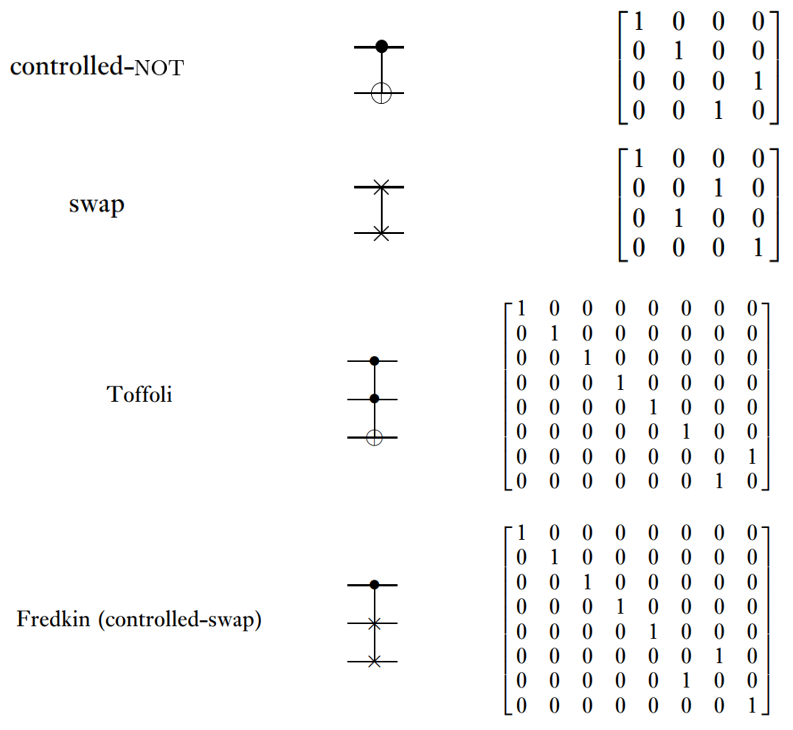}} 
  \caption{Multiple-qubit quantum gates}
   \label{fig:MQG}
\end{figure}

In an ion-trap quantum computer, CNOT gates can only operate serially. Even though different CNOT gates involve different qubits, they cannot operate in parallel \cite{yang2013post}. Therefore, the number of CNOT gates significantly affects the running time of quantum algorithms. When considering the computational complexity of quantum algorithms, we focus on the number of CNOT gates. Since CNOT gates are very critical in quantum computers, in order to calculate the number of CNOT gates in the subsequent article, we decompose a Toffoli gate into single-qubit quantum gates and CNOT gates.

Shende and Markov confirmed \cite{shende2009cnot} that no matter whether the auxiliary system is used or not, at least 6 CNOT gates are needed to implement the Toffoli gate decomposition. A Toffoli gate can be decomposed into six CNOT gates, seven T gates, two H gates, and one S gate. The quantum circuit is as Figure \ref{fig:DOTTG}.
\begin{figure}[htbp]
    \centering
    \includegraphics[scale=0.28]{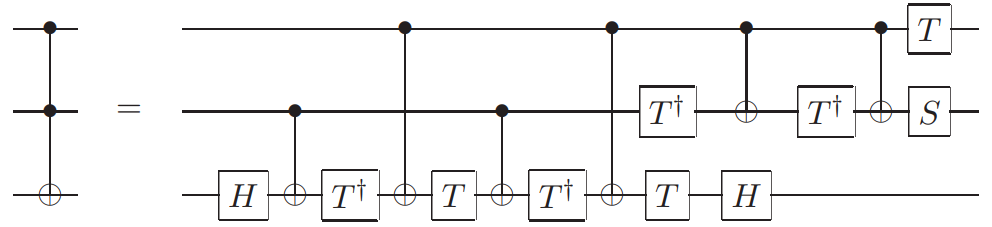}
    \caption{Decomposition of the Toffoli Gate}
    \label{fig:DOTTG}
\end{figure}

In \cite{DBLP:journals/corr/abs-2304-03050}, seven CNOT gates can be used to implement the decomposition of Fredkin gates. A Fredkin gate can be decomposed into seven CNOT gates, seven T gates, two H gates, and three S gates. The quantum circuit is Figure \ref{fig:DOTFG}.
\begin{figure}[htbp]
    \centering
    \includegraphics[scale=0.55]{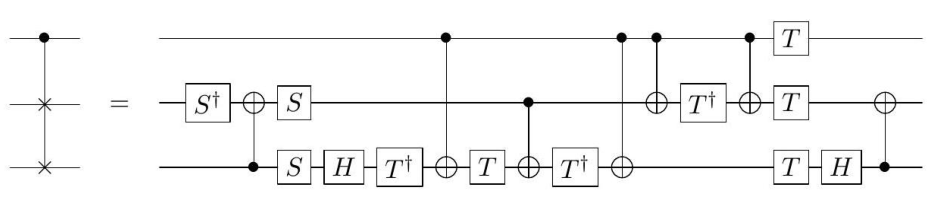}
    \caption{Decomposition of the Fredkin Gate}
    \label{fig:DOTFG}
\end{figure}

Usually, in order to reduce the loss of resources, Toffoli gates are usually used to implement the exchange effect of Fredkin gates, so as to reduce the number of decomposed CNOT gates, such as Algorithm \ref{Alg:QAFGSAR} and Algorithm \ref{Alg:QAFGSARAV} in section \ref{Algorithms}.

\section{Algorithms for Solving Quantum Linear Systems of Equations}  \label{Algorithms}
In quantum linear equations, there is no need to distinguish column pivots and full pivots (because the pivot can only be 1, there is no statement of selecting the largest pivot, and it does not involve precision issues). Only the precise general algorithm needs to be considered, and we analyze it based on the ideas of Gaussian elimination and Gaussian-Jordan elimination.

We divided the section into three subsections. In section 3.1, we first translate the classical Gaussian-Jordan elimination algorithm over binary fields to the equivalent quantum algorithm, so as to analyze the number of quantum gates required. In section 3.2, we propose a quantum algorithm using the quantum data register itself as a storage register, and the number of quantum gates for solving the quantum linear equations with coherent superposition can be polynomially reduced. In section 3.3, we propose another quantum algorithm using a quantum auxiliary register as a storage register. Although the number of auxiliary qubits for solving the quantum linear equations with coherent superposition increases, this method can be disentangled with other registers and keep the quantum data register unchanged.

\begin{definition}
\textbf{(Quantum linear equations with coherent superposition)} We define a quantum state matrix $|O\rangle$ to be the tensor product of $mn$ $|0\rangle$, i.e. $|O\rangle=|0\rangle^{\otimes mn}$. After Hadamard transformation $H^{\otimes mn}$, denote $H^{\otimes mn}|O\rangle=|A\rangle$, then
\begin{equation*}
\begin{aligned}
    |A\rangle &=\left(\frac{1}{\sqrt{2}}(|0\rangle+|1\rangle)\right)^{\otimes mn} =(\frac{1}{\sqrt{2}})^{mn}\sum_{a_{ij}\in \mathbb{F}_2}|a_{11}\cdots a_{1n}a_{21}\cdots a_{2n}\cdots a_{m1}\cdots a_{mn}\rangle.
\end{aligned}
\end{equation*}

If measured, $|A\rangle$ will collapse to a certain classical matrix $A$ over binary fields. It can be written as:
\begin{equation*}
\begin{bmatrix} 
\frac{1}{\sqrt{2}}(|0\rangle+|1\rangle)&\frac{1}{\sqrt{2}}(|0\rangle+|1\rangle)&\cdots&\frac{1}{\sqrt{2}}(|0\rangle+|1\rangle)\\ 
\frac{1}{\sqrt{2}}(|0\rangle+|1\rangle)&\frac{1}{\sqrt{2}}(|0\rangle+|1\rangle)&\cdots&\frac{1}{\sqrt{2}}(|0\rangle+|1\rangle)\\ 
\vdots&\vdots&\ddots&\vdots\\
\frac{1}{\sqrt{2}}(|0\rangle+|1\rangle)&\frac{1}{\sqrt{2}}(|0\rangle+|1\rangle)&\cdots&\frac{1}{\sqrt{2}}(|0\rangle+|1\rangle)
\end{bmatrix}_{m\times n}
\stackrel{\text{Collapse}}{\longrightarrow} 
\begin{bmatrix} 
a_{11}&a_{12}&\cdots&a_{1n}\\ 
a_{21}&a_{22}&\cdots&a_{2n}\\ 
\vdots&\vdots&\ddots&\vdots\\
a_{m1}&a_{m2}&\cdots&a_{mn}
\end{bmatrix}_{m\times n},
\end{equation*} $a_{ij} \in \mathbb{F}_2$, i.e., $a_{ij}$ is 0 or 1. 

Given a  quantum state constant vector
$|b\rangle_{m\times 1} =|b_1\rangle \otimes \cdots \otimes |b_m\rangle$, it can be written as: $\begin{bmatrix} 
|b_1\rangle\\ 
|b_2\rangle\\ 
\vdots\\
|b_m\rangle
\end{bmatrix}_{m\times 1}, b_i \in \mathbb{F}_2.$

If there is a quantum state variable vector $|x\rangle$ that collapses to $x$, such that $Ax=b$, then $|x\rangle$ is called the solution of this quantum linear equations with coherent superposition $|A\rangle |x\rangle=|b\rangle$.
\end{definition}

\subsection{Quantum Gaussian-Jordan elimination for general solution and rank}
\begin{theorem}\label{theorem3}
\textbf{(Quantum Turing completeness theorem \cite{nielsen2010quantum})} A quantum computer can theoretically solve any problem that a classical computer can solve, and it can outperform classical computers in terms of speed when tackling certain problems.
\end{theorem}

According to Theorem \ref{theorem3}, we first translate the classical Gaussian-Jordan elimination Algorithm \ref{Alg:GEFGS} in the Appendix to the equivalent quantum algorithm, as shown in Algorithm \ref{Alg:QGEFGS}. 
\begin{algorithm}[H]
		\caption{\textbf{Quantum Gaussian-Jordan elimination for general solution and rank}}
		\label{Alg:QGEFGS}
		\begin{algorithmic}[1]
			\Require 
		    $|[A|b]\rangle$: quantum coefficient augmented matrix, where $|A\rangle$ is a $m\times n$ matrix 
			\Ensure  $\exists$ $|x\rangle$ collapses into $x$ such that $Ax=b$
            \State   \textbf{if} $m<n$ \textbf{then}
			\State \ \ \ \ add $(n-m)\times (n+1)$ and $n\times (n+1)$ all-zero quantum state below the original matrix
			\State   \textbf{else} 
			\State \ \ \ \ no add $(n-m)\times (n+1)$ but add $n\times (n+1)$ all-zero quantum state below the original matrix
			\State   $Pivot_{11}=0$; $\cdots$; $Pivot_{1n}=0$; $\cdots$; $Pivot_{n1}=0$; $\cdots$; $Pivot_{nn}=0$;
			\State \textbf{for}  $i \leftarrow 1$ $\textbf{to}$ $m$ \textbf{do}
		    \State \ \ \ \  \textbf{for} $j\leftarrow 1 $ $\textbf{to}$ $n$ \textbf{do}
	        \State \ \ \ \ \ \ \ \  $work_1=0$; $\cdots$; $work_{j-1}=0$; $Toffoli(X(a_{i1});X(a_{i2});work_1)$;
	        \State \ \ \ \ \ \ \ \   $Toffoli(X(a_{i3});work_1;work_2)$; $\cdots$; $Toffoli(X(a_{i(j-1)});work_{j-3};work_{j-2})$; 
	        \State \ \ \ \ \ \ \ \ $Toffoli(a_{ij};work_{j-2};work_{j-1})$; $CNOT(work_{j-1};Pivot_{ij})$;
			\State \ \ \ \ \ \ \ \    \textbf{for} $\ell \leftarrow 1$  $\textbf{to}$ $i-1$ and $\ell \leftarrow i+1$  $\textbf{to}$ $m$ \textbf{do}
			\State \ \ \ \ \ \ \ \ \ \ \ \ \ \ \ \   $radd_{\ell j}=0$ ; $Toffoli(Pivot_{ij}; a_{\ell j}; radd_{\ell j})$ ;
		    \State \ \ \ \ \ \ \ \ \ \ \ \ \ \ \ \     $Toffoli(radd_{\ell j};a_{i:};a_{\ell:})$ ; $Toffoli(radd_{\ell j};b_{i};b_{\ell})$ ;
		    \State \ \ \ \ \ \ \ \  $Fredkin(Pivot_{ij};a_{i:},a_{j:})$; $Fredkin(Pivot_{ij};b_{i},b_{j})$;
		    \State   \textbf{for}  $j \leftarrow 1$ $\textbf{to}$ $n$ \textbf{do}
			\State \ \ \ \  $Toffoli(a_{jj};b_{j};b_{(j+n)})$ ; 
			\State \ \ \ \   \textbf{for}  $i \leftarrow 1$ $\textbf{to}$ $j-1$ and $i \leftarrow j+1$ $\textbf{to}$ $n$  \textbf{do}
			\State \ \ \ \ \ \ \ \ $X(a_{jj})$ ; $Toffoli(a_{jj};a_{ij};a_{(i+n)j})$ ; $X(a_{jj})$ ;
			\State \ \ \ \ $X(a_{jj})$ ; $CNOT(a_{jj};a_{(j+n)j})$; $X(a_{jj})$ ;
		    \State   \textbf{for}  $j \leftarrow 1$ $\textbf{to}$ $n$ \textbf{do}
			\State \ \ \ \  $solution_j=0$ ;
			\State \ \ \ \  \textbf{for}  $h \leftarrow 1$ $\textbf{to}$ $n$ \textbf{do}
			\State \ \ \ \ \ \ \ \ $k_{h}=0$ ; $Toffoli(H(k_{h});a_{(j+n)h};solution_j)$ ;
			\State \ \ \ \  $CNOT(b_{j+n};solution_j)$ ;
			\State \ \ \ \  $x_j=solution_j$ ;
			\State  \textbf{return} $x=(x_1,\cdots,x_n)$ ; $rank(A)=count(a_{jj}==1)$
		\end{algorithmic}
\end{algorithm}
\noindent Note: $X$ denotes the NOT gate; $H$ denotes the Hadamard gate; $CNOT(a;b)$ denotes $b$ $\oplus$$=a$; $Toffoli(a;b;c)$ denotes $c$ $\oplus$$=ab$; $Fredkin(a;b,c)$ indicates that when $a=1$, $swap(b;c)$. $a_{i:}$ represents all elements (the quantum state of 0 or 1) in the $i$-th row. $work_j$ indicates the auxiliary qubits in the process of finding the pivot (when $j=1$, perform $CNOT(a_{i1};Pivot_{i1})$; when $j=2$, perform $Toffoli(X( a_{i1});a_{i2};Pivot_{i2})$; when $j\geq 3$, perform steps 8-10), $Pivot_{ij}$ indicates the auxiliary qubits that are used to judge whether it is the pivot column for each column of each row, $k_h$ indicates the auxiliary qubits that are used to construct the coefficients of the basic solution system, $radd$ and $solution$ indicate the auxiliary qubits that need to be used in the row addition and solution process respectively.

Brief description of Algorithm \ref{Alg:QGEFGS}: If $m<n$, add $(n-m)\times(n+1)$ zero quantum states below the original augmented matrix for Gaussian elimination transformation and a $n\times (n+1)$ all-zero matrix for storing the basic solution system and special solution. First traverse by row, then traverse by column, and use Toffoli gates to find the pivot in each row. If $a_{ij}$ is the first element with 1 in the $i$-th row, mark it as the pivot with a CNOT gate and store it in auxiliary qubits $Pivot_{ij}$; use the $i$-th row to eliminate the row where the element is 1 (except the pivot) in the pivot column. Using $Pivot_{ij}$ as control qubits, exchange the elements of the $i$-th row and the $j$-th row such that the pivot is on the main diagonal. After traversing all the rows and columns, if a column is the pivot column, then the data qubit $a_{jj}$ is 1; otherwise, $a_{jj}$ is 0. Continue to traverse by column, if the $j$-th column is the column where the pivot is located, then use the Toffoli gates to XOR $b_j$ to $b_{j+n}$ ($j=1,\cdots,n$), and store it in the special solution; otherwise, use the Toffoli gates to XOR $a_{ij }$ to $a_{(i+n)j}$ ($i,j=1,\cdots,n$) and assign $a_{(j+n)j}$ to 1, then store them in the basic solution system. Then the coefficient before each vector in the basic solution system is obtained to be 0 or 1 by using the Hadamard transformation. By measuring, $x_j=b_j+k_1a_{(j+n)1}+\cdots+k_na_{(j +n)n}$.

We analyze the number of quantum gates needed by quantum Gaussian-Jordan elimination algorithm for solving quantum linear equations. Steps 8-10 need $m(n-1)$ CNOT gates, step 19 needs $n$ CNOT gates, step 24 needs $n$ CNOT gates, the total number of CNOT gates is $mn+2n-m$; steps 8-10 need $\frac{(n-1)mn}{2}$ Toffoli gates, step 12 needs $(m-1)mn$ Toffoli gates, step 13 needs $(m-1)mn(n+1)$ Toffoli gates, step 16 needs $n$ Toffoli gates, step 18 needs $(n-1)n$ Toffoli gates, step 23 needs $n^2$ Toffoli gates, the total number of Toffoli gates is $\frac{2m^2n^2+4m^2n-mn^2+4n^2-5mn}{2}$; step 14 needs $mn(n+1)$ Fredkin gates. Therefore, the number of CNOT gates required is $\mathcal{O}(mn)$; the number of Toffoli gates required is $\mathcal{O}(m^2n^2)$; the number of Fredkin gates required is $\mathcal{O}(mn^2)$.

\subsection{Quantum algorithm for general solution and rank}
In this section, we propose a quantum algorithm that given m n-qubit vectors as input, computes the rank of its span and general solutions. The algorithm can polynomially reduce the number of quantum gates compared with Algorithm \ref{Alg:QGEFGS}.

This algorithm needs to add $n\times (n+1)$ zero quantum states as auxiliary qubits above the original quantum state matrix to form a $(m+n)\times (n+1)$ quantum state matrix, which does not change the rank and solution of the original quantum linear system of equations. We can use the original matrix only as the quantum storage register to complete the process of solving the general solution and rank with this quantum algorithm. We use the symbols $tag_i$ and $mark_j$ (which can be regarded as auxiliary qubits). $tag_i$ indicates whether the $i$-th row in the original quantum state matrix is added to the $j$-th row of the all-zero matrix; $mark_j$ indicates whether the $j$-th column is the pivot column or the free variable column after adding the zero quantum states as auxiliary qubits. 

When $m<n$, the $n\times (n+1)$ matrix added above the original matrix is used to store the quantum state matrix $|U\rangle$ after being transformed into an upper triangle; $(n-m)\times (n+1)$ matrix $O$ below the original matrix is stored in the same register as the original matrix. According to Theorem \ref{theorem1}, considering the case where there are solutions, we analyze from the coefficient augmented matrix (ignoring the amplitude here), then the coefficient augmented matrix can be transformed into the form as in formula (\ref{eq1}).
 \begin{equation}
 \label{eq1}
|[A|b]\rangle_{m \times (n+1)} \stackrel{\text{Add zero qubits}}{\longrightarrow} |A'\rangle=\begin{bmatrix} 
|0\rangle&|0\rangle&\cdots&|0\rangle&|0\rangle\\ 
\vdots&\vdots&\ddots&\vdots&\vdots \\
|0\rangle&|0\rangle&\cdots&|0\rangle&|0\rangle\\
|a_{(n+1)1}\rangle&|a_{(n+1)2}\rangle&\cdots&|a_{(n+1)n}\rangle&|b_1\rangle\\
\vdots&\vdots&\ddots&\vdots&\vdots\\
|a_{(m+n)1}\rangle&|a_{(m+n)2}\rangle&\cdots&|a_{(m+n)n}\rangle&|b_m\rangle\\
|0\rangle&|0\rangle&\cdots&|0\rangle&|0\rangle\\ 
\vdots&\vdots&\ddots&\vdots&\vdots\\
|0\rangle&|0\rangle&\cdots&|0\rangle&|0\rangle\\
\end{bmatrix}_{2n\times (n+1)} 
\begin{matrix}
\square &  |mark_1\rangle \\
\vdots &\\
\square &\ |mark_n\rangle \\
\square &\ |tag_1\rangle \\
\vdots &\\
\square &\ |tag_m\rangle 
 &  \\
 &  \\
 &  \\
 &  \\
\end{matrix}
\end{equation}

When $m\geq n$, only need to add $n\times (n+1)$ quantum states $O$ above the original matrix to store the quantum state matrix $|U\rangle$ transformed into the upper triangle, and its coefficient augmented matrix is transformed into the form as in formula (\ref{eq2}).
\begin{equation}
\label{eq2}
|[A|b]\rangle_{m \times (n+1)} \stackrel{\text{Add zero qubits}}{\longrightarrow} |A''\rangle=\begin{bmatrix} 
|0\rangle&|0\rangle&\cdots&|0\rangle&|0\rangle\\ 
\vdots&\vdots&\ddots&\vdots&\vdots\\
|0\rangle&|0\rangle&\cdots&|0\rangle&|0\rangle\\
|a_{(n+1)1}\rangle&|a_{(n+1)2}\rangle&\cdots&|a_{(n+1)n}\rangle&|b_1\rangle\\
\vdots&\vdots&\ddots&\vdots&\vdots\\
|a_{(m+n)1}\rangle&|a_{(m+n)2}\rangle&\cdots&|a_{(m+n)n}\rangle&|b_m\rangle\\
\end{bmatrix}_{(m+n)\times (n+1)} 
\begin{matrix}
\square &\ |mark_1\rangle \\
\vdots &\\
\square &\ |mark_n\rangle \\
\square &\ |tag_1\rangle \\
\vdots &\\
\square &\ |tag_m\rangle 
\end{matrix}
\end{equation}

In the above two formulas, $|a_{ij}\rangle$ in the original matrix corresponds to $|a_{(i+n)j}\rangle$ in the new matrix, $a_{ij} \in \mathbb{F}_2$ ; $|b_{i}\rangle$ in the original matrix corresponds to $|b_{(i+n)}\rangle$ in the new matrix, $b_{i} \in \mathbb{F}_2$. Initialize $tag_i$ and $mark_j$, set all $tag_i$ to 0 and all $mark_j$ to 1. The detailed quantum algorithm is as shown in Algorithm \ref{Alg:QAFGSAR}.

Brief description of Algorithm \ref{Alg:QAFGSAR}: add $n \times (n+1)$ zero quantum states above the original matrix, when $m<n$, the $(n-m)\times (n+1)$ zero quantum states need to be added to the original matrix register; when $m \geq n $, there is no need to add the zero quantum states. After traversing all the rows and columns, store the basic solution system and special solution in the original matrix register. The first $n \times n$ qubits store the basic solution system, and $n \times 1$ qubits store the special solution. Use the intermediate variables $tag_{i-n}$ (initialized to 0) and $mark_j$ (initialized to 1) for marking. Traverse the elements of the $j$-th column, if $a_ {ij}$ ($i=n+1, \cdots, 2n$) is the first element with 1 in the $j$-th column, then update the value of $tag_{i-n}$ to 1, the value of $mark_j$ to 0, and add the $i$-th row to the $j$-th row to achieve swapping with Toffoli gates; continue to find the element $a_{\ell j}$ ($\ell=1,\cdots,j-1, n+1,\cdots,m+n$) of the $1$-th to the $(j-1)$-th row and the $(n+1)$)-th to the $(m+n)$-th row, $tag_{i-n}$ remains unchanged and is still 0, $mark_j$ is still 0 after updating; use the $j$-th row to eliminate the row where the element in the $\ell$-th column is 1 with Toffoli gates. Judge by the updated value of $mark_j$, if the value of $mark_j$ is 1, then this column is the column where the free variable is located. Add this column to the $j$-th column of the register where the original matrix is located, i.e., the $(n+1)$-th row to the $(2n)$-th row of the $j$-th column of the new matrix, and assign $a_{(j+n)j}$ to 1; if the value of $mark_j$ is 0, then this column is the column where the pivot is located, and the $b_j$ is the special solution corresponding to the $j$-th row. This special solution is added to the $(n+1)$-th element of the $(n+j)$-th row of the new matrix. The $(n+1)$-th to $(2n)$-th rows of the new matrix correspond to the solutions $x_1,\cdots,x_n$ respectively, and the auxiliary qubits $|k_h\rangle$ are subjected to Hadamard transformation, so that the coefficient of the basic solution system is 0 or 1, then add its row to the auxiliary qubits $|solution\rangle$.

\begin{algorithm}[htb!]
		\caption{\textbf{Quantum algorithm for general solution and rank}}
		\label{Alg:QAFGSAR}
		\begin{algorithmic}[1]
			\Require 
		    $|[A'|b]\rangle$ ($|[A''|b]\rangle$) is a coefficient augmented matrix belongs to $\mathbb{F}^{m\times(n+1)}_2$, where $|A'\rangle$ is a $(2n)\times (n+1)$ matrix and $|A''\rangle$ is a $(m+n)\times (n+1)$ matrix 
			\Ensure  $\exists$ $|x\rangle$ collapses into $x$ such that $Ax=b$
			\State   \textbf{if} $m<n$ \textbf{then}
			\State \ \ \ \  \textbf{Operate} $|A'\rangle$
            \State \ \ \ \  $mark_1=X(0)$; $\cdots$ ; $mark_n=X(0)$;
	        \State \ \ \ \   $tag_1=0$;  $\cdots$ ; $tag_m=0$;
			\State \ \ \ \  \textbf{for} $j \leftarrow 1$ $\textbf{to}$ $n$ \textbf{do}
			\State \ \ \ \ \ \ \ \ \textbf{for}  $i \leftarrow n+1$ $\textbf{to}$ $m+n$ \textbf{do}
			\State \ \ \ \ \ \ \ \ \ \ \ \ $Toffoli(a_{ij};mark_j;tag_{i-n})$;$Toffoli(a_{ij};tag_{i-n};mark_j)$;
			\State \ \ \ \ \ \ \ \ \ \ \ \  
            $Toffoli(tag_{i-n};a_{i:};a_{j:})$;  $Toffoli(tag_{i-n};b_{i};b_{j})$;
			\State \ \ \ \ \ \ \ \ \textbf{for}  $\ell \leftarrow 1$ $\textbf{to}$ $j-1$ and $\ell \leftarrow n+1$ $\textbf{to}$ $m+n$  \textbf{do}
			\State \ \ \ \ \ \ \ \ \ \ \ \ $radd_{\ell j}=0$ ;  $CNOT(a_{\ell j};radd_{\ell j})$;
			\State \ \ \ \ \ \ \ \ \ \ \ \  $Toffoli(radd_{\ell j};a_{j:};a_{\ell:})$ ; $Toffoli(radd_{\ell j};b_{j};b_{\ell})$ ;
			\State \ \ \ \ \ \ \ \   \textbf{for}  $k \leftarrow n+1$ $\textbf{to}$ $n+j-1$ and $k \leftarrow n+j+1$ $\textbf{to}$ $2n$ \textbf{do}
			\State \ \ \ \ \ \ \ \ \ \ \ \ $Toffoli(mark_{j};a_{(k-n)j};a_{kj})$ ;
			\State \ \ \ \ \ \ \ \  $CNOT(mark_{j};a_{(j+n)j})$ ;
			\State \ \ \ \ \ \ \ \  $X(mark_{j})$  ; $Toffoli(mark_{j};b_{j};b_{j+n})$ ; $X(mark_{j})$  ;
			\State \ \ \ \  \textbf{for}  $j \leftarrow 1$ $\textbf{to}$ $n$ \textbf{do}
			\State \ \ \ \ \ \ \ \ $solution_j=0$ ;
			\State \ \ \ \ \ \ \ \ \textbf{for}  $h \leftarrow 1$ $\textbf{to}$ $n$ \textbf{do}
			\State \ \ \ \ \ \ \ \ \ \ \ \ $k_{h}=0$ ; $Toffoli(H(k_{h});a_{(j+n)h};solution_j)$ ;
			\State \ \ \ \ \ \ \ \  $CNOT(b_{j+n};solution_j)$ ;
			\State \ \ \ \ \ \ \ \ $x_j=solution_j$ ;
			\State   \textbf{else} 
			\State \ \ \ \  \textbf{Operate} $|A''\rangle$
            \State \ \ \ \  repeat steps 3-21
			\State  \textbf{return} $x=(x_1,\cdots,x_n)$ ; $rank(A)=count(mark_j==0)$
		\end{algorithmic}
\end{algorithm}

\noindent Note: "$k \leftarrow n+j+1$ $\textbf{to}$ $2n$" can be omitted in step 12, but it does not affect the order of magnitude of the number of quantum gates. Symbol description is the same as Algorithm \ref{Alg:QGEFGS}.

Its quantum circuit is as shown in Figure \ref{fig:QCFSLE}:
\begin{figure}[H]
    \centering
    \begin{subfigure}{0.9\linewidth}
        \includegraphics[scale=0.2]{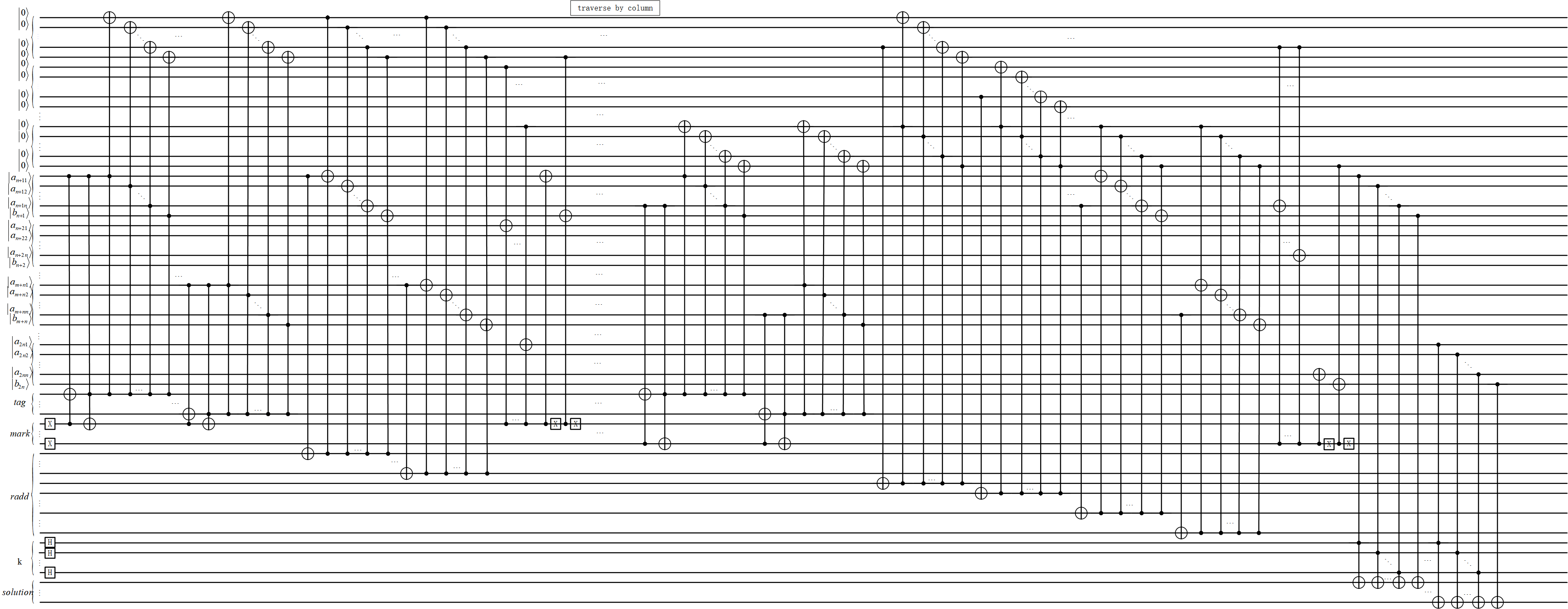}
        \caption{Quantum circuit for solving linear equation when $m<n$.}
        \label{fig:QCFSLE_a}
    \end{subfigure}
    \begin{subfigure}{0.9\linewidth}
        \includegraphics[scale=0.2]{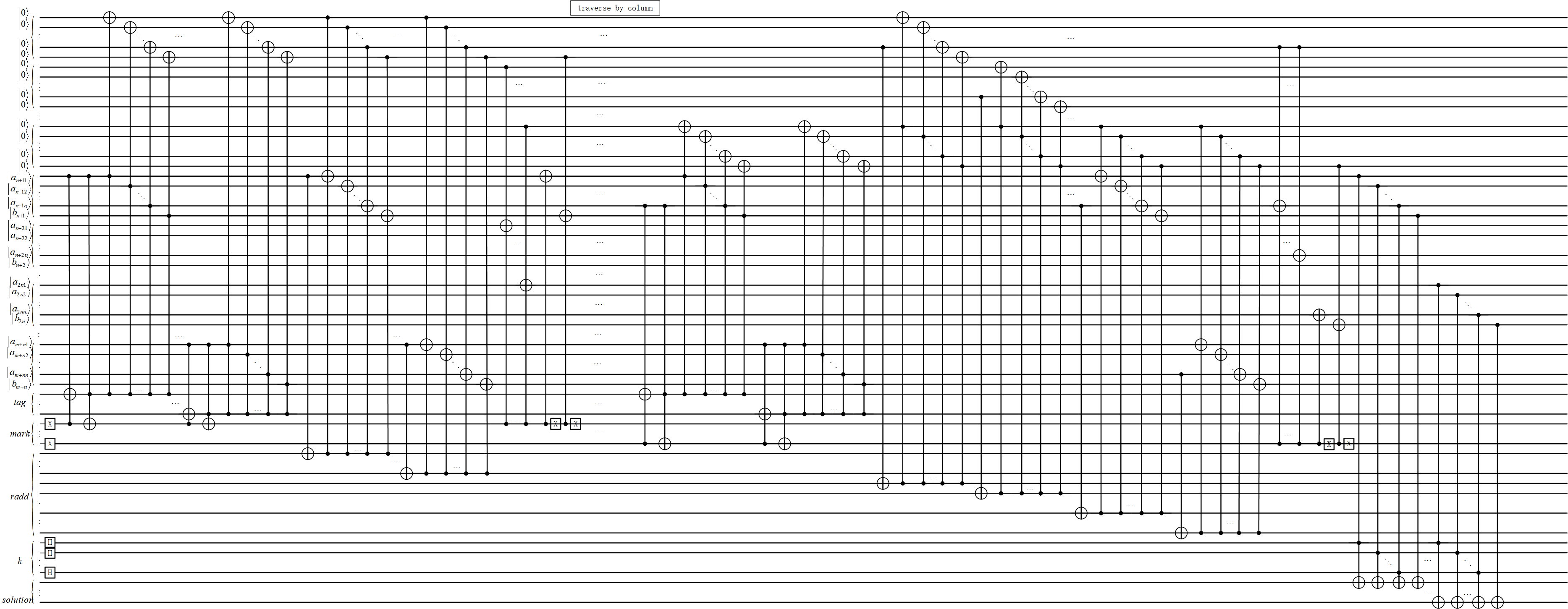}
        \caption{Quantum circuit for solving linear equation when $m\geq n$.}
        \label{fig:QCFSLE_b}
    \end{subfigure}
    \caption{Quantum circuit for solving linear equation }
    \label{fig:QCFSLE}
\end{figure}

We analyze the number of quantum gates needed by the quantum algorithm Algorithm \ref{Alg:QAFGSAR} for solving the quantum linear equations. Step 10 needs $mn+\frac{n(n-1)}{2}$ CNOT gates, step 14 needs $n$ CNOT gates, step 20 needs $n$ CNOT gates, the total number of CNOT gates is $\frac{2mn+n^2+3n}{2}$; step 7 needs $2mn$ Toffoli gates, step 8 needs $mn(n+1)$ Toffoli gates, step 11 needs $mn(n+1)+\frac{n(n-1)(n+1)}{2}$ Toffoli gates, step 13 needs $n(n-1)$ Toffoli gates, step 15 needs $n$ Toffoli gates, step 19 needs $n^2$ Toffoli gates, the total number of Toffoli gates is $\frac{4mn^2+n^3+8mn+4n^2-n}{2}$. Therefore, the number of CNOT gates required is $\mathcal{O}(\frac{2mn+n^2}{2})$; the number of Toffoli gates required is $\mathcal{O}(\frac{4mn^2+n^3}{2})$.

More visually, we do not write other registers. when $m<n$, $|A'\rangle$ is operated, and $|A'\rangle$ will become the following form after Algorithm \ref{Alg:QAFGSAR}:
\begin{equation*}
    \begin{aligned}
      |A'\rangle =
        \begin{bmatrix} 
            |O\rangle_{n\times n}&|O\rangle_{n\times 1}\\ 
            |\frac{A_{m\times n}}{O_{(n-m)\times n}}\rangle_{n \times n}&|\frac{b_{m\times 1}}{O_{(n-m)\times 1}}\rangle_{n\times 1}\\ 
        \end{bmatrix}_{2n \times (n+1)} 
    \stackrel{\text{Algorithm 2}}{\longrightarrow}
        \begin{bmatrix} 
            |U\rangle_{n\times n}&|b'\rangle_{n\times 1}\\ 
            |\eta\rangle_{n\times n}&|b'\rangle_{n\times 1}
        \end{bmatrix}_{2n\times (n+1)},   
    \end{aligned}
\end{equation*}
where $|\frac{A_{m\times n}}{O_{(n-m)\times n}}\rangle_{n \times n}$ indicates that the all-zero quantum state matrix $|O\rangle_{(n-m)\times n}$ is added below the original matrix $|A\rangle_{m \times n}$, $|\frac{b_{m\times 1}}{O_{(n-m)\times 1}}\rangle_{n\times 1}$ indicates that all-zero vector $|O\rangle_{(n-m)\times 1}$ is added below the constant vector $|b\rangle_{m \times 1}$. $|b'\rangle_{n\times 1}$ is a special solution vector, $|\eta\rangle_{n\times n}$ is a basic solution system.

When $m\geq n$, $|A''\rangle$ is operated, and $|A''\rangle$ will become the following form after Algorithm \ref{Alg:QAFGSAR}:
\begin{equation*}
    \begin{aligned}
|A''\rangle =
\begin{bmatrix} 
|O\rangle_{n\times n}&|O\rangle_{n\times 1}\\ 
|A\rangle_{m \times n}&|b\rangle_{m\times 1}\\ 
\end{bmatrix}_{(m+n) \times (n+1)} 
\stackrel{\text{Algorithm 2}}{\longrightarrow}
\begin{bmatrix} 
|U\rangle_{n\times n}&|b'\rangle_{n\times 1}\\ 
|\frac{\eta_{n\times n}}{O_{(m-n)\times n}}\rangle_{m \times n}&|\frac{b'_{n\times 1}}{O_{(m-n)\times 1}}\rangle_{m\times 1}\\ 
\end{bmatrix}_{(m+n)\times (n+1)}, 
    \end{aligned}
\end{equation*}
where $|\frac{\eta_{n\times n}}{O_{(m-n)\times n}}\rangle_{m \times n}$ indicates that the all-zero quantum state matrix $|O\rangle_{(m-n)\times n}$ is added below the basic solution system $|\eta\rangle_{n\times n}$, $|\frac{b'_{n\times 1}}{O_{(m-n)\times 1}}\rangle_{m\times 1}$ indicates that the all-zero vector $|O\rangle_{(m-n)\times 1}$ is added below the special solution vector $|b'\rangle_{n \times 1}$. $|U\rangle_{n\times n}$ is an upper triangular quantum state matrix, $|b'\rangle_{n\times 1}$ a special solution vector.

Its universal quantum circuit is as shown in Figure \ref{fig:GQAFGSAR}.
\begin{figure}[H]
    \centering
    \includegraphics[scale=0.28]{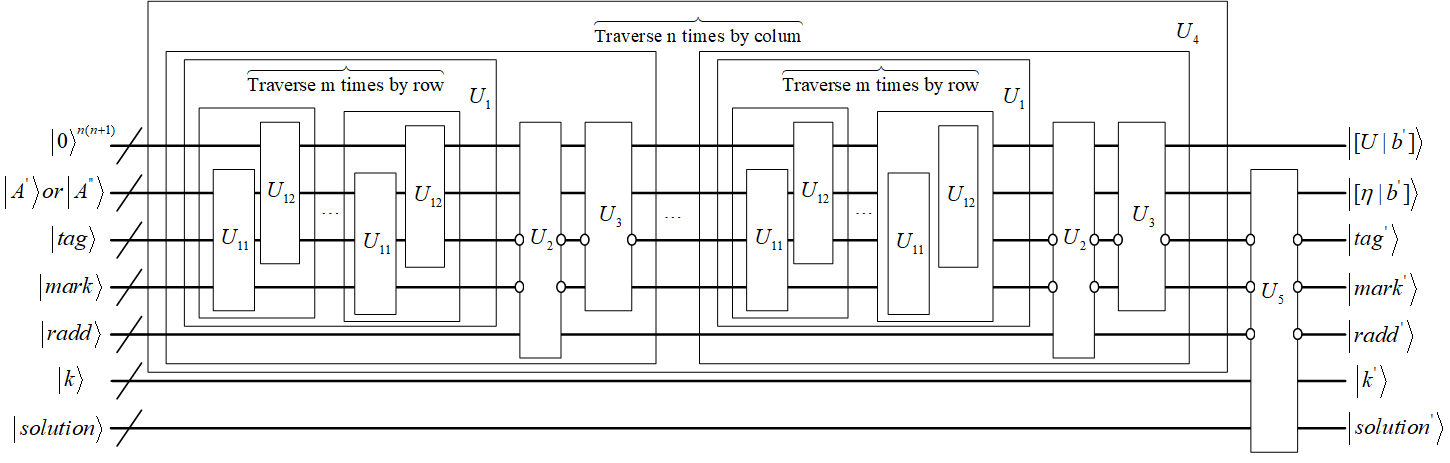}
    \caption{Universal quantum circuit for general solution and rank}
    \label{fig:GQAFGSAR}
\end{figure}

Note that $|A'\rangle_{2n\times (n+1)}$ and $|A''\rangle_{(m+n)\times (n+1)}$ are operated respectively when $m<n$ and $m \geq n$. $U_i$ means unitary transformation, and the registers are entangled with eacn other after $U_i$. "$\circ$" means not to undergo this unitary transformation. In Algorithm \ref{Alg:QAFGSAR}, for each column, step 7 can be represented by $U_{11}$, step 8 can be represented by $U_{12}$, steps 6-8 can be represented by $U_1$, judging whether the row is added and whether the column is the pivot column; steps 9-11 can be represented by $U_2$, eliminating the rest elements with 1 in the pivot column; steps 12-15 can be represented by $U_3$, judging whether the colum is a free variable column or a pivot column and operating respectively. Steps 5-15 can be represented by $U_4$, storing the general solution; steps 16-21 can be represented by $U_5$, obtaining the final solution (by last measurement).

This quantum algorithm cannot disentangle the data register containing $|A'\rangle$ or $|A''\rangle$ with other auxiliary registers since the data register is used as control qubits, and the unitary operation $U_5$ is performed. But auxiliary qubits can be saved, and they can be directly used in the circuits of some quantum algorithms for solving linear equations. For example, it can be directly applied to the circuit of parallel Simon's algorithm (in section \ref{Application}). However, in the case where we need to obtain the solutions of the quantum linear equations by one measurement at last, the data register can be used as a single input and output, and the data register is disentangled with other auxiliary registers at this time.

\subsection{Quantum algorithm for general solution and rank (Another Version)}
We propose another algorithm that uses an auxiliary register as a storage register. This algorithm needs to add $n\times (n+1)$ zero quantum states as auxiliary qubits above and below the original quantum state matrix to form a $(m+2n)\times (n+1)$ quantum state matrix, which does not change the rank and solution of the original quantum linear equations. Similarly, the symbols $tag_i$ and $mark_j$ are also given. The difference from Algorithm \ref{Alg:QAFGSAR} is that we use the $n\times (n+1)$ zero quantum states added below as the storage register for storing the general solution. It can completely disentangle the original quantum coefficient matrix register with other registers and keep data qubits unchanged. According to Theorem \ref{theorem1}, considering the case where there are solutions, we analyze from coefficient augmented matrix (ignoring the amplitude), the transformation form is as shown in formula (\ref{eq3}), where $|a_{ij}\rangle$ and $|b_{i}\rangle$ of the original matrix correspond to $|a_{(i+n)j}\rangle$ and $|b_{( i+n)}\rangle$, respectively.
 \begin{equation}
 \label{eq3}
|[A|b]\rangle_{m \times (n+1)} \stackrel{\text{Add zero qubits}}{\longrightarrow} |A'''\rangle=\begin{bmatrix} 
|0\rangle&|0\rangle&\cdots&|0\rangle&|0\rangle\\ 
|0\rangle&|0\rangle&\cdots&|0\rangle&|0\rangle\\ 
\vdots&\vdots&\ddots&\vdots&\vdots\\
|0\rangle&|0\rangle&\cdots&|0\rangle&|0\rangle\\
|a_{(n+1)1}\rangle&|a_{(n+1)2}\rangle&\cdots&|a_{(n+1)n}\rangle&|b_{n+1}\rangle\\
\vdots&\vdots&\ddots&\vdots&\vdots\\
|a_{(m+n)1}\rangle&|a_{(m+n)2}\rangle&\cdots&|a_{(m+n)n}\rangle&|b_{n+m}\rangle\\
|0\rangle&|0\rangle&\cdots&|0\rangle&|0\rangle\\ 
|0\rangle&|0\rangle&\cdots&|0\rangle&|0\rangle\\ 
\vdots&\vdots&\ddots&\vdots&\vdots\\
|0\rangle&|0\rangle&\cdots&|0\rangle&|0\rangle\\
\end{bmatrix}_{(m+2n)\times (n+1)} 
\begin{matrix}
\square & \ |mark_1\rangle \\
\square & \ |mark_2\rangle \\
\vdots &\\
\square & \ |mark_n\rangle \\
\square & \ |tag_1\rangle \\
\vdots &\\
\square & \ |tag_m\rangle 
 &  \\
 &  \\
 &  \\
 &  \\
 &  \\
\end{matrix}
\end{equation}

The $n\times (n+1)$ matrix added above the original matrix is used to store the quantum state matrix $|U\rangle$ after the upper triangle; the $n\times (n+1)$ matrix added below the original matrix is used to store the general solution. The first $n\times n$ quantum qubits store the basic solution system, and the latter $n\times 1$ quantum qubits store the special solution. The detailed quantum algorithm is as shown in Algorithm \ref{Alg:QAFGSARAV}.

\begin{algorithm}[H]
		\caption{\textbf{Quantum algorithm for general solution and rank (Another Version)}}
		\label{Alg:QAFGSARAV}
		\begin{algorithmic}[1]
			\Require 
		    $|[A'''|b]\rangle$ is a coefficient augmented matrix belongs to $\mathbb{F}^{m\times(n+1)}_2$, where $|A'''\rangle$ is a $(m+2n)\times (n+1)$ matrix 
			\Ensure  $\exists$ $|x\rangle$ collapses into $x$ such that $Ax=b$
            \State   $mark_1=X(0)$; $\cdots$ ; $mark_n=X(0)$;
	        \State    $tag_1=0$;  $\cdots$ ; $tag_m=0$;
			\State   \textbf{for} $j \leftarrow 1$ $\textbf{to}$ $n$ \textbf{do}
            \State \ \ \ \  run steps 6-11 of algorithm 2
			\State \ \ \ \   \textbf{for}  $k \leftarrow m+n+1$ $\textbf{to}$ $m+n+j-1$ and $k \leftarrow m+n+j+1$ $\textbf{to}$ $m+2n$ \textbf{do}
			\State \ \ \ \ \ \ \ \  $Toffoli(mark_{j};a_{(k-(m+n))j};a_{kj})$ ;
			\State \ \ \ \ $CNOT(mark_{j};a_{(j+m+n)j})$ ;
			\State \ \ \ \   $X(mark_{j})$;$Toffoli(mark_{j};b_j;b_{j+m+n})$;$X(mark_{j})$;
			\State   \textbf{for}  $j \leftarrow 1$ $\textbf{to}$ $n$ \textbf{do}
			\State \ \ \ \  $solution_j=0$ ;
			\State \ \ \ \  \textbf{for}  $h \leftarrow 1$ $\textbf{to}$ $n$ \textbf{do}
			\State \ \ \ \ \ \ \ \ $k_{h}=0$;$Toffoli(H(k_{h});a_{(j+m+n)h};solution_j)$;
			\State \ \ \ \  $CNOT(b_{j+m+n};solution_j)$;
			\State \ \ \ \  $x_j=solution_j$;
			\State  \textbf{return} $x=(x_1,\cdots,x_n)$; $rank(A)=count(mark_j==0)$
		\end{algorithmic}
\end{algorithm}
\noindent Note: Same as Algorithm 2, "$k \leftarrow m+n+j+1$ $\textbf{to}$ $m+2n$" can be omitted in step 5, but it does not affect the order of magnitude of the number of quantum gates. Symbol description is the same as Algorithm \ref{Alg:QGEFGS}.

Its quantum circuit is as shown in Figure \ref{fig:QCFSLEAV}:
\begin{figure*}[htbp]
    \centering
    \includegraphics[scale=0.2]{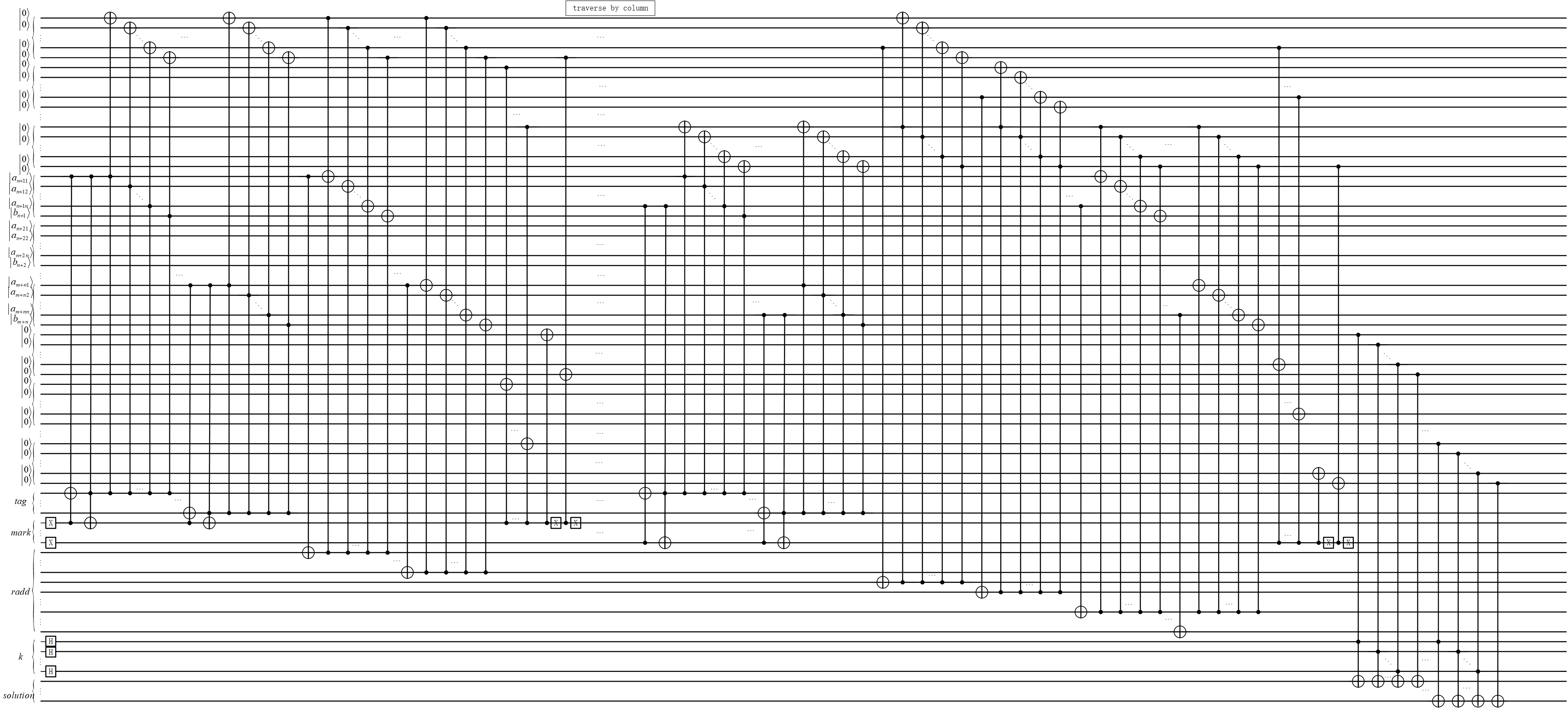}
    \caption{Quantum circuit for solving linear equation (Another Version)}
    \label{fig:QCFSLEAV}
\end{figure*}

Brief description of Algorithm \ref{Alg:QAFGSARAV}: The whole process is similar to Algorithm \ref{Alg:QAFGSAR}, the difference is that not only adding $n \times (n+1)$ zero quantum states above the original matrix, but also adding $n \times (n+1)$ zero quantum states below the original matrix for storing the general solution. Similarly, the $(m+n+1)$-th to the $(m+2n)$-th rows of the new matrix correspond to the solutions $x_1,\cdots,x_n$ respectively, and the auxiliary qubits $|k_h\rangle$ are subjected to Hadamard transformation such that the coefficients of the basic solution system is 0 or 1, then add its row to the auxiliary qubits $|solution\rangle$.

We analyze that Algorithm \ref{Alg:QAFGSARAV} and Algorithm \ref{Alg:QAFGSAR} need the same number of quantum gates for solving quantum linear equations.

Vividly, $|A'''\rangle$ will become the following form after Algorithm \ref{Alg:QAFGSARAV}:
\begin{equation*}
    \begin{aligned}
    |A'''\rangle =
\begin{bmatrix} 
|O\rangle_{n\times n}&|O\rangle_{n\times 1}\\ 
|A\rangle_{m \times n}&|b\rangle_{m\times 1}\\ 
|O\rangle_{n\times n}&|O\rangle_{n\times 1}
\end{bmatrix}_{(m+2n)\times (n+1)} 
\stackrel{\text{Algorithm 3}}{\longrightarrow}
\begin{bmatrix} 
|U\rangle_{n\times n}&|b'\rangle_{n\times 1}\\ 
|O\rangle_{m \times n}&|O\rangle_{m\times 1}\\ 
|\eta\rangle_{n\times n}&|b'\rangle_{n\times 1}
\end{bmatrix}_{(m+2n)\times (n+1)}. 
    \end{aligned}
\end{equation*}

Each row of the quantum state matrix is a register. Here, $|A\rangle$ represents the quantum coefficient matrix, $|b\rangle$ represents the quantum constant vector, $|O\rangle$ represents the all-zero quantum state matrix, $|U\rangle$ represents the quantum upper triangular matrix and the elements of the pivot column are zero except for the pivot, $|\eta\rangle$ represents the quantum basic solution system, and $|b'\rangle$ represents the quantum special solution vector. Its universal quantum circuit is shown in Figure \ref{fig:GQAFGSARAV}.
\begin{figure}[H]
    \centering
    \includegraphics[scale=0.45]{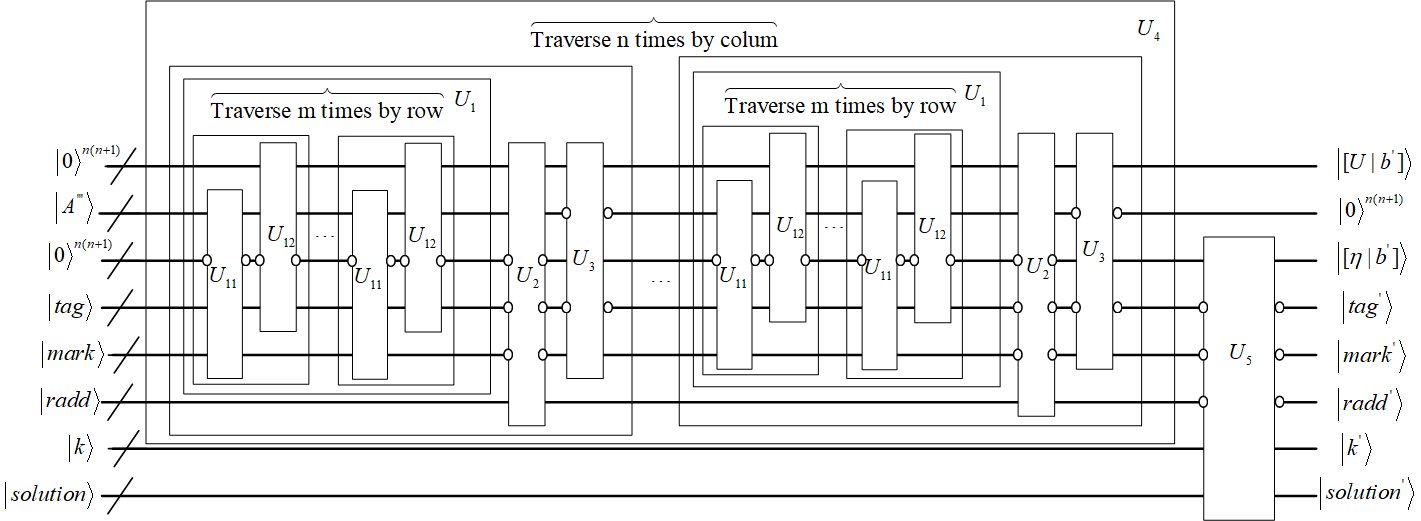}
    \caption{Universal quantum circuit for general solution and rank (Another Version)}
    \label{fig:GQAFGSARAV}
\end{figure}

Unlike Figure \ref{fig:GQAFGSAR}, there is an additional $n\times(n+1)$-qubit auxiliary register for storing the general solution. This is used to disentangle the data register where the original matrix is located with other registers and keep data qubits unchanged, so as to facilitate subsequent applications.

\vspace{0.25cm}
\noindent \textbf{Remark} In summary, we compare Algorithm \ref{Alg:QAFGSAR} and Algorithm \ref{Alg:QAFGSARAV}. If we need to measure the final result, we do not need to add a third register, i.e., we do not need to add the $n\times (n+1)$ zero quantum states below the original matrix for storing the general solution. Directly use algorithm \ref{Alg:QAFGSAR} to measure the data register, then we can get a $n\times n(n+1)$ matrix $[\eta|b']$ that stores the basic solution system and special solution. Compared with Algorithm $\ref{Alg:QAFGSARAV}$, $n(n+1)$ qubits are saved. If the solutions of quantum linear equations need to be measured at last, the basic solution system and special solution can be directly output in the data register, and they can be disentangled with other auxiliary qubits (not add the registers containing $|k\rangle$ and $|solution\rangle$); if there is no solution, it can also be judged by the quantum storage register. It is as shown in following Figure \ref{fig:GQAFBSSASS}.
\begin{figure}[H]
    \centering
    \includegraphics[scale=0.65]{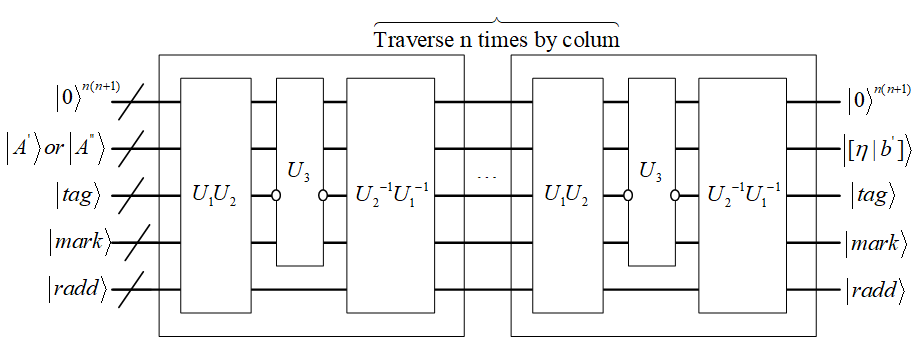}
    \caption{Universal quantum circuit for basic solutions system and special solutions (a)}
    \label{fig:GQAFBSSASS}
\end{figure}
The unitary operators $U_1, U_2, U_3$ in Figure $\ref{fig:GQAFBSSASS}$ correspond to $U_1, U_2, U_3$ in Figure $\ref{fig:GQAFGSAR}$. Auxiliary registers' inputs are the same as outputs and they are disentangled with the data registers.

Due to the nature of quantum entanglement, different registers cannot be measured simultaneously. Therefore, we must store the general solution in the auxiliary register. If we do not need to measure the final result, we need to use the auxiliary register for subsequent other applications. At this time, it is necessary to apply Algorithm \ref{Alg:QAFGSARAV}, which can reverse the first two registers, and only keep the general solution in the auxiliary register added, as shown in Figure \ref{fig:GQAFBSSASSAV}.
\begin{figure}[H]
    \centering
    \includegraphics[scale=0.63]{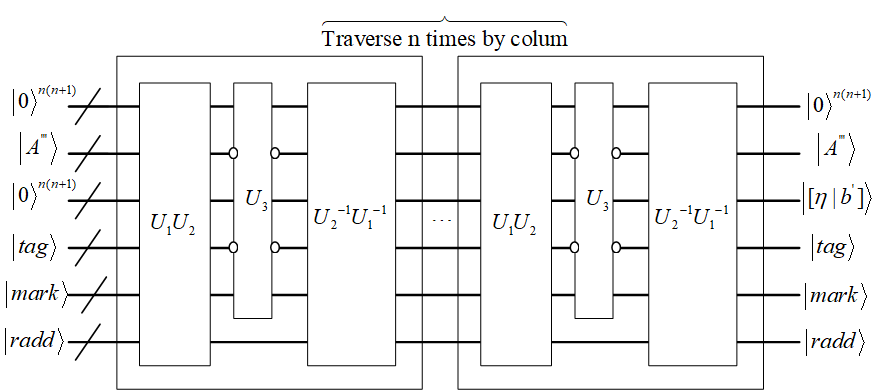}
    \caption{Universal quantum circuit for basic solutions system and special solutions (b)}
    \label{fig:GQAFBSSASSAV}
\end{figure}
In Figure $\ref{fig:GQAFBSSASSAV}$, $U_1,U_2,U_3$ all correspond to $U_1,U_2,U_3$ in Figure $\ref{fig:GQAFGSARAV}$. The data register is disentangled with other auxiliary registers and the data qubits are unchanged.

\section{Application} \label{Application}
In this section, we apply Algorithm \ref{Alg:QAFGSAR} and Algorithm \ref{Alg:QAFGSARAV}  as a subroutine to parallel Simon's algorithm \cite{simon1997power} (with multiple periods), Grover Meets Simon algorithm \cite{leander2017grover}, and Alg-PolyQ2 Algorithm \cite{bonnetain2019quantum}. The desired results can be obtained by one measurement at last.

\subsection{Application to parallel Simon's Algorithm}
The goal of Simon's algorithm \cite{simon1997power} is to solve the periodic problem of hidden subgroups, which can be described as the following problem:

\textbf{Simon's Problem:} Suppose there is a quantum oracle with a computable function $f$: $\{0,1\}^n \xrightarrow{}\{0,1\}^n$, and it satisfies the promise: $\exists$ $s \in \{0,1\}^n$, such that $\forall$ $x \in \{0,1\}^n$, we have $f(x)=f(x\oplus s)$. The goal is to find the period $s$.

Simon's algorithm can be performed by the following steps:
\begin{enumerate}
    \item Prepare the initial state$|0\rangle^{\otimes n}_{\uppercase\expandafter{\romannumeral1}}|0\rangle^{\otimes n}_{\uppercase\expandafter{\romannumeral2}}$, and perform the Hadamard transform $H^{\otimes n}$ on the first register. The following quantum states can be obtained:
    \[ \frac{1}{\sqrt{2^n}} \sum_{x \in \{0,1\}^n} |x\rangle_{\uppercase\expandafter{\romannumeral1}} |0\rangle_{\uppercase\expandafter{\romannumeral2}}.  \]
    \item Perform a quantum query on the function $f$, mapping to the quantum state:
    \[ \frac{1}{\sqrt{2^n}} \sum_{x \in \{0,1\}^n} |x\rangle_{\uppercase\expandafter{\romannumeral1}} |f(x)\rangle_{\uppercase\expandafter{\romannumeral2}}.  \]
    \item Perform the Hadamard transform $H^{\otimes n}$ on the first register, and the following quantum states can be obtained:
     \[ \frac{1}{2^n} \sum_{y \in \{0,1\}^n}\sum_{x \in \{0,1\}^n} (-1)^{x\cdot y}|y\rangle_{\uppercase\expandafter{\romannumeral1}} |f(x)\rangle_{\uppercase\expandafter{\romannumeral2}}.  \]
     Suppose there is some $s \neq 0$, for each $y$, then $|y,f(x)\rangle=|y,f(x\oplus s)\rangle$, and the amplitude of this configuration is:
     \begin{equation}
     \label{equ1}
        \frac{1}{2^n} [(-1)^{x\cdot y}+(-1)^{(x\oplus s)\cdot y}] = \frac{1}{2^n} (-1)^{x\cdot y}[(1+(-1)^{s\cdot y}].
     \end{equation}
      \item 
      The amplitude satisfying $y \cdot s=1$ is 0. Therefore, it will randomly and uniformly collapse after measuring the first register, and we can obtain a value of y, which must satisfy $y \cdot s=0$.
\end{enumerate}

Repeat the above algorithm $\mathcal{O}(n)$ times to get $n-1$ linearly independent values of $y$, then $s$ can be obtained by solving the following classical linear equations:
\begin{equation}
\label{equ2}
    \left\{ 
\begin{aligned}
    &y_1 \cdot s =0 \\
    &y_2 \cdot s =0 \\
    &\cdots \\
    &y_n \cdot s =0 .\\
\end{aligned}
\right.
\end{equation}

Now we apply Algorithm \ref{Alg:QAFGSAR} in section \ref{Algorithms} to  parallel Simon's algorithm. Since Simon's algorithm is probabilistic, we parallel $m=\mathcal{O}(n)$ Simon's algorithms and finally obtain the period $s$ with high probability by one measurement. The quantum circuit that applies Algorithm \ref{Alg:QAFGSAR} for solving the quantum linear equations to  parallel Simon's algorithm is as shown in Figure \ref{fig:QLEASA}.
\begin{figure}[H]
    \centering
    \includegraphics[scale=0.6]{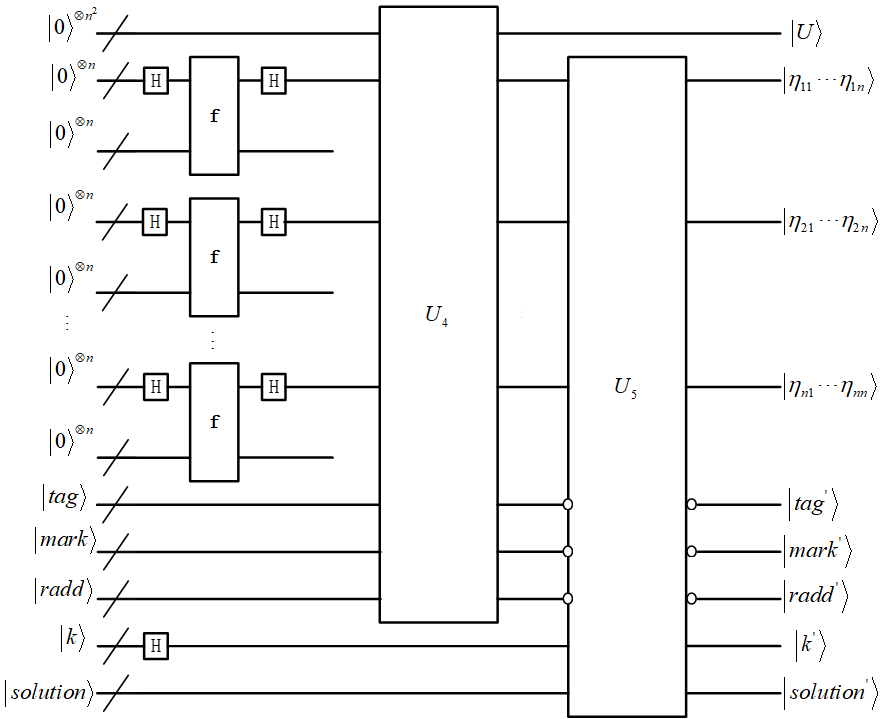}
    \caption{Quantum linear equations applied parallel Simon's algorithm}
    \label{fig:QLEASA}
\end{figure}

Next, we compute the probability that we can obtain the solution by one measurement.

In Figure \ref{fig:QLEASA}, there are a total of $m=\mathcal{O}(n)$ Simon's algorithms in parallel, and the quantum states obtained by the first register of each Simon's algorithm are stored together in $|A''\rangle$ of Algorithm \ref{Alg:QAFGSAR} (since it is a homogeneous quantum linear system of equations, there is no need to use coefficient augmented matrix, i.e., $|A''\rangle$ is a $m\times n$ matrix here).

\begin{lemma}\label{lemma1}
Suppose $Y \subset \mathbb{F}_2^n$ is an $(n-1)$-dimensional subspace, randomly select $y_1,\cdots,y_{n-1}\in Y$ at uniform , then the probability that $y_1,\cdots,y_{n-1}$ is linearly independent is at least 0.288.
\end{lemma}

\begin{proof}
For the first register in each Simon's algorithm of the parallel Simon's algorithm, the probability that $y \cdot s=0$ is
\begin{equation*}
    \begin{aligned}
        Pr[y|y\cdot s=0] = \sum_{x \in R}|\frac{1}{2^n} (-1)^{x\cdot y}(1+(-1)^{s\cdot y})|^2 
        = \frac{1}{2^{n+1}} |1+(-1)^{s\cdot y}|^2,
    \end{aligned}
\end{equation*}
where $R$ is a subgroup over binary fields and is the coset representation, $|R|=2^{n-1}$, then $Pr[y|s\cdot y=0]=\frac{1}{2^{n-1}}$. If the set of Eqs.(\ref{equ2}) has only one nontrivial solution, according to Theorem \ref{theorem2}, when we query $n$ times, we can get $n-1$ linearly independent vectors to obtain the only non-zero solution, i.e., the period $s\neq 0$, and the probability at this time is
\begin{equation*}
    \begin{aligned}
    Pr[linearly \quad independent|s\cdot y=0]
    = \frac{2^{n-1}-1}{2^{n-1}} \cdot  \frac{2^{n-1}-2}{2^{n-1}} \cdots  \frac{2^{n-1}-2^{n-2}}{2^{n-1}} 
    =\prod_{i=1}^{n-1} (1-\frac{1}{2^i}) \geq \prod_{i=1}^{\infty} (1-\frac{1}{2^i}). 
    \end{aligned}
\end{equation*}
According to Euler's pentagon theorem, $Pr[linearly \quad independent|s\cdot y=0]= \prod_{i=1}^{n-1} (1-\frac{1}{2^i}) \geq \prod_{i=1}^{\infty} (1-\frac{1}{2^i})\approx 0.288$.
\end{proof}

Thereby, we know that Simon's algorithm is a probabilistic algorithm, and the promise of Simon's problem is not always fully satisfied in cryptanalysis. In order to improve the success probability of Simon's algorithm as much as possible, we need to increase the number of queries. Next, a theorem is given to weigh the relationship between the parallel number of Simon's queries and the success probability of Simon's algorithm. This theorem can be found similarly in \cite{kaplan2016breaking}:
\begin{theorem}\label{theorem4}
Given $f:\{0,1\}^n \xrightarrow{} \{0,1\}^n$, s is a period of f, $\exists$ $a\in \{0,1\}^n \backslash \{0,s\}$, such that the input $x\in \{0,1\}^n$ satisfies $f(x\oplus a)=f(x)$. Define
 \begin{equation}
 \label{equ3}
      \varepsilon(f)=\max_{a\in \{0,1\}^n \backslash
 \{0,s\}} Pr_x[f(x)=f(x\oplus a)].  
 \end{equation}
 If $\varepsilon(f)\leq p_0<1$, after performing m Simon's algorithms in parallel, the success probability of outputing the period $s$ is at least $1-2^n(\frac{1+p_0}{2})^m$. When choosing $m\geq 3n/(1-p_0)$, it can ensure that $\varepsilon(f)$ is far enough away from 1 in eq.(\ref{equ3}), so as to ensure that the success probability of Simon's algorithm is as large as possible.
\end{theorem}

From Lemma \ref{lemma1}, we can see that the upper bound $p_0$ of $\varepsilon(f)$ is $1-0.288=0.712$. According to Theorem \ref{theorem4}, we only need to ensure that the number of parallel Simon's algorithm is $m\geq 3n/(1- p_0)\approx 10.4n$, measure the $|solution'\rangle$ register once at last, then the solution (period $s$) to the problem can be obtained with a probability close to 1.

Simon's quantum algorithm can also be applied to the case where the Boolean function F has two or more periods, i.e., there are multiple non-zero solutions to the quantum linear system of equations.

\textbf{Simon's Problem with multiple periods:} Suppose there is a quantum oracle with a computable function $f$: $\{0,1\}^n \xrightarrow{}\{0,1\}^n$. It satisfies the promise: there exist linearly independent vectors $s_1,\cdots,s_k \in \{0,1\}^n$, such that $\forall$ $x \in \{0,1\}^n$, we have $f(x)=f(x\oplus s_i)$. The goal is to find the periods $\{s_1,\cdots,s_k\}$.

Perform steps 1-3 of Simon's algorithm, then the amplitude of the resulting quantum state is
\begin{equation}
\begin{aligned}
    \label{equ4}
      \frac{1}{2^n} (-1)^{x\cdot y}[1+(-1)^{s_1\cdot y}][1+(-1)^{s_2\cdot y}]\cdots[1+(-1)^{s_k\cdot y}]. 
\end{aligned}
\end{equation}

Since the amplitude of $y$ satisfying $y \cdot s_i=1$ is 0, there are m blocks (totally $m\times k$ equations), and we can divide them into k groups to obtain $s_1,\cdots,s_k$:
 \begin{equation}
     \label{equ5}
     \left\{ 
\begin{aligned}
    &y_1 \cdot s_1 =0 \\
    &y_2 \cdot s_1 =0 \\
    &\cdots \\
    &y_m \cdot s_1 =0 \\
\end{aligned}
\right.
,
\cdots
,
\left\{ 
\begin{aligned}
    &y_1 \cdot s_k =0 \\
    &y_2 \cdot s_k =0 \\
    &\cdots \\
    &y_m \cdot s_k =0. \\
\end{aligned}
\right.
 \end{equation}
Eqs.(\ref{equ5}) is actually the same as Simon's original Eqs.(\ref{equ2}).

\begin{lemma}\label{lemma2}
Suppose $Y \subset \mathbb{F}_2^n$ is an $(n-k)$-dimensional subspace, randomly select $y_1,\cdots,y_{n-k}\in Y$ at uniform, then the probability that $y_1,\cdots,y_{n-k}$ is linearly independent is at least 0.288.
\end{lemma}

\begin{proof}
Let $E$ be the event that $(y\cdot s_1=0)\land\cdots\land(y\cdot s_k=0)\land(s_i,s_j \text{\ is linearly independent} (i,j=1,\cdots,k,i\neq j))$, same as the proof of Lemma \ref{lemma1}, for the first register in each Simon's algorithm of the parallel Simon's algorithm, the probability satisfying $E$ is
\begin{equation*}
    \begin{aligned}
         Pr[y|E] = \sum_{x \in R}|\frac{1}{2^n} (-1)^{x\cdot y}[1+(-1)^{s_1\cdot y}]\cdots[1+(-1)^{s_k\cdot y}]|^2 
         = \frac{1}{2^{n+k}} |[1+(-1)^{s_1\cdot y}]\cdots[1+(-1)^{s_k\cdot y}]|^2,
    \end{aligned}
\end{equation*}
where $R$ is a subgroup over binary fields and is the coset representation, $|R|=2^{n-k}$, then $Pr[y|E]=\frac{1}{2^{n-k}}$. According to Theorem \ref{theorem2}, if we query $n$ times, we can get $n-k$ linearly independent vectors to obtain $k$ non-zero solutions, the probability at this time is
\begin{align*}
    Pr[linearly \quad independent|E]
    = \frac{2^{n-k}-1}{2^{n-k}} \cdot  \frac{2^{n-k}-2}{2^{n-k}} \cdots  \frac{2^{n-k}-2^{n-k-1}}{2^{n-k}} 
    =\prod_{i=1}^{n-k} (1-\frac{1}{2^i}) \geq \prod_{i=1}^{\infty} (1-\frac{1}{2^i}).
\end{align*}

Similarly, $Pr[linearly \quad independent|E]= \prod_{i=1}^{n-1} (1-\frac{1}{2^i}) \geq \prod_{i=1}^{\infty} (1-\frac{1}{2^i})\approx 0.288$.
\end{proof}

Unlike the original Simon's algorithm, which only has a non-zero period, the last measured register is the data register where $|A''\rangle$ is located (shown in \ref{fig:GQAFBSSASS}), rather than the $|solution'\rangle$ register. This is because, if we measure the solution register $|solution'\rangle$, it will only collapse to one of the solutions (period $s_i$), but we can get the basic solution system and special solutions by measuring the register where $|A''\rangle$ is located, so as to obtain all solutions (periods $\{s_1,\cdots,s_k\}$).

Similarly, according to Theorem \ref{theorem4}, as long as the number of parallel Simon's algorithm is guaranteed $m\geq 3n/(1-p_0)\approx 10.4n$,
then the solution (all periods $\{s_1,\cdots,s_k\}$) to the problem can be obtained with a probability close to 1 by measuring $|A''\rangle$ register.

\subsection{Application to Grover Meets Simon Algorithm}
The Grover Meets Simon algorithm \cite{leander2017grover} performs a quantum search, uses Simon's algorithm to identify the correct guess, and can attack the FX-construction. It solves the problem of searching for a periodic function, which can be described as the following problem:

$\textbf{Grover Meets Simon Problem:}$ Suppose there is a quantum oracle with a computable function $f$: $\{0,1\}^m \times  \{0,1\}^n \xrightarrow{}\{0,1\}^n$, $\exists$ unique $k_0 \in \{0,1\}^m$, such that $\forall$ $x \in \{0,1\}^n$, we have $f(k_0,x)=f(k_0,x\oplus k_1)$. The goal is to find $k_0$ and period $k_1$.

$\textbf{Simon’s promise:}$ Since some function values $f(k_0,x)$ may have more than two preimages, the function is not a mapping of $2 \xrightarrow{} 1$. Suppose that $u_i$ is a vector in the linear space of periodic function $f(k_0,x)$, choose $\ell=2(n+\sqrt{n})$, so that $\langle u_1,\cdots,u_\ell\rangle$ span a linear space of $f(k_0,x)$ with high probability (at least $\frac{4}{5}$). At this time, under the assumption that $f(k_0,x)$ is a random periodic function whose period is $k_1$, the probability that any function value $f(k_0,x)$ has only two preimages is at least $\frac{ 1}{2}$.

We analyze the Grover Meets Simon algorithm in more detail, reconstruct the classifier $\mathcal{B}$ in the original algorithm from the view of quantum, and while ensuring the high probability of solving $k_0$ and period $k_1$, Algorithm \ref{Alg:QAFGSAR} for solving linear equations in section \ref{Algorithms} is applied to Grover Meets Simon algorithm as a subroutine. Moreover, we construct the specific quantum circuit of Grover Meets Simon algorithm.

Here we first show the universal quantum circuit of Grover Meets Simon algorithm:
\begin{figure}[H]
    \centering
    \includegraphics[scale=0.5]{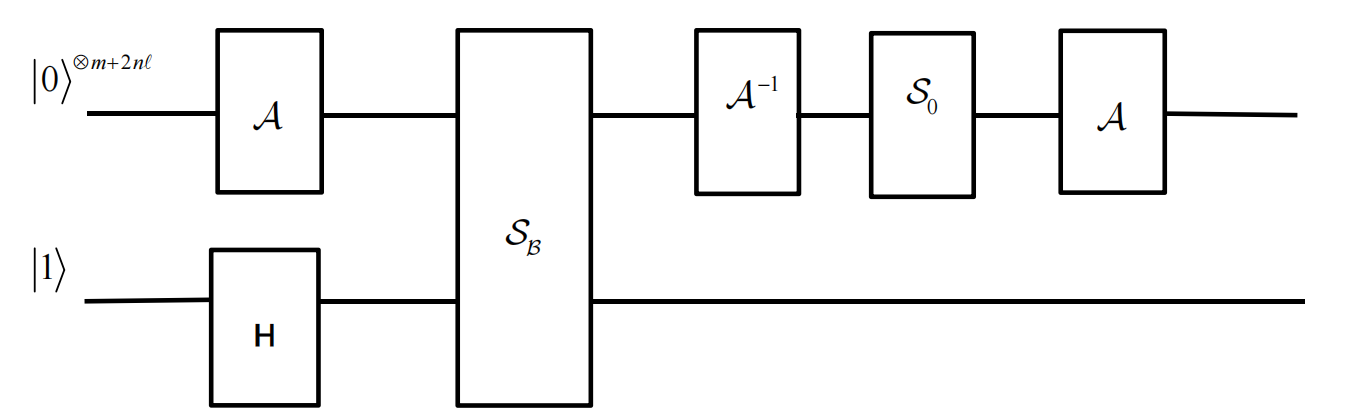}
    \caption{Universal quantum circuit of Grover Meets Simon algorithm}
    \label{fig:GQCOGMS}
\end{figure}

In fact, the quantum algorithm $\mathcal{A}$ is a parallel Simon's algorithm with control qubits $|k_0\rangle$ and works on the input $|0\rangle^{\otimes (m+2n\ell)}$, the steps are as follows:
\begin{enumerate}
     \item  Prepare the initial state $|0\rangle^{\otimes m}_{\text{\uppercase\expandafter{\romannumeral1}}}|0\rangle^{\otimes n\ell}_{\text{\uppercase\expandafter{\romannumeral2}}}|0\rangle^{\otimes n\ell}_{\text{\uppercase\expandafter{\romannumeral3}}}$.
     \item  Perform Hadamard transform $H^{\otimes m+n \ell}$ on the first $m+n \ell$ qubits, the result is
     \[ \frac{1}{\sqrt{2^{m+n\ell}}}\sum_{\substack{k \in \mathbb{F}_2^m\\ x_1, \cdots, x_\ell \in \mathbb{F}_2^n}} |k \rangle_{\text{\uppercase\expandafter{\romannumeral1}}} 
     \left(|x_1 \rangle  \cdots |x_\ell \rangle \right)_{\text{\uppercase\expandafter{\romannumeral2}}} |0 \rangle^{\otimes n\ell}_{\text{\uppercase\expandafter{\romannumeral3}}}. \]
     \item Perform a unitary transformation $U_h$, the result is
     \begin{equation*}
         \begin{aligned}
         \frac{1}{\sqrt{2^{m+n\ell}}}\sum_{\substack{k \in \mathbb{F}_2^m \\ x_1, \cdots, x_\ell \in \mathbb{F}_2^n}} |k \rangle_{\text{\uppercase\expandafter{\romannumeral1}}}  
        (|x_1 \rangle  \cdots &|x_\ell \rangle)_{\text{\uppercase\expandafter{\romannumeral2}}} |h(k,x_1,\cdots, x_\ell) \rangle_{\text{\uppercase\expandafter{\romannumeral3}}}. 
         \end{aligned}
     \end{equation*}
     \item Perform Hadamard transform on the $(m+1)\text{-th},\cdots ,(m+n\ell)\text{-th}$ qubit, the result is
     \begin{equation*}
         \begin{aligned}
             |\psi \rangle = \frac{1}{\sqrt{2^{m+n\ell}}}\sum_{\substack{k \in \mathbb{F}_2^m, u_1, \cdots, u_\ell \in \mathbb{F}_2^n \\ x_1, \cdots, x_\ell \in \mathbb{F}_2^n}} |k \rangle_{\text{\uppercase\expandafter{\romannumeral1}}} \left((-1)^{\langle u_1, x_1 \rangle} |u_1 \rangle  \right. 
            \left. \cdots (-1)^{\langle u_\ell, x_\ell \rangle} |u_\ell \rangle\right)_{\text{\uppercase\expandafter{\romannumeral2}}}  |h(k,x_1,\cdots, x_\ell) \rangle_{\text{\uppercase\expandafter{\romannumeral3}}}.   
         \end{aligned}
     \end{equation*}
\end{enumerate}

Assume that measure the last $n\ell$ qubits of $|\psi \rangle$ in step 4, then these qubits will collapse to
 \[ |h(k,x_1,\cdots, x_\ell) \rangle = |f(k,x_1)||f(k,x_2)||\cdots ||f(k,x_\ell) \rangle, \]
 for some fixed $k,x_1,\cdots , x_\ell \in \mathbb{F}_2^n$.

An arbitrary $n$-qubit state $|z_i \rangle = (-1)^{\langle u_i, x_i \rangle} |u_i \rangle$ from $\psi$ is entangled with $|f(k_0,x_i) \rangle$. Therefore, after measuring $|f(k_0,x_i) \rangle$, $|z_i \rangle$ collapses into a superposition state.
     
If $x_i$ and $x_i+k_1$ are the only two preimages of $f(k_0,x_i) \rangle$, then $|z_i \rangle$ is called a proper state. And $|z_i \rangle$ collapses to the superposition
\begin{equation}
\begin{aligned}
   \label{equ6}
    \left( (-1)^{\langle u_i, x_i \rangle}\right.  + \left. (-1)^{\langle u_i, x_i+k_1 \rangle} \right) |u_i \rangle 
     = (-1)^{\langle u_i, x_i \rangle}\left(1+(-1)^{\langle u_i, k_1 \rangle} \right) |u_i \rangle,  
\end{aligned}
\end{equation}
it can be seen from eq.(\ref{equ6}) that $|u_i \rangle$ has a non-vanishing amplitude if and only if $\langle u_i, k_1 \rangle =0$.

Here, we reconstruct the classifier $\mathcal{B}$ (in Grover Meets Simon algorithm) as a new quantum classifier $\mathcal{S_B}$, then we give its description and quantum circuit construction based on Algorithm \ref{Alg:QAFGSAR} (in fact, this quantum construction can also be applied to the oracle construction in other quantum algorithms):
 
\begin{enumerate}
\item Based on the quantum state $|\psi\rangle$, we apply Algorithm \ref{Alg:QAFGSAR} to construct the specific quantum circuit of Grover Meets Simon algorithm. Use the second register as control qubits to perform $U_4|\psi\rangle$, and store the value in the fourth register. We can obtain the following result (other auxiliary qubits are not listed, the same below):
 \begin{equation*}
 \begin{aligned}
     \frac{1}{\sqrt{2^{m+n\ell}}} \sum_{\substack{k \in \mathbb{F}_2^m, u_1, \cdots, u_\ell \in \mathbb{F}_2^n \\ x_1, \cdots, x_\ell \in \mathbb{F}_2^n}} |k \rangle_{\text{\uppercase\expandafter{\romannumeral1}}} \left((-1)^{\langle u_1, x_1 \rangle}  \cdots \right. (-1)^{\langle u_\ell, x_\ell \rangle} 
     \left.|\eta_1 \cdots \eta_n 0^{(\ell-n)n} \rangle\right)_{\text{\uppercase\expandafter{\romannumeral2}}}  |h(k,x_1,\cdots, x_\ell) \rangle_{\text{\uppercase\expandafter{\romannumeral3}}} 
 |mark \rangle_{\text{\uppercase\expandafter{\romannumeral4}}},
 \end{aligned}
 \end{equation*}
where $|\eta_1 \cdots \eta_n\rangle$ is the basic solution system of quantum linear equations $|u_1 \cdots u_\ell \rangle$.

\item  Use the second and the fourth register as control qubits to perform $U_5$ and store the value in the fifth register. We can obtain the following result
\begin{equation*}
 \begin{aligned}
     \frac{1}{\sqrt{2^{m+n\ell}}} \sum_{\substack{k \in \mathbb{F}_2^m, u_1, \cdots, u_\ell \in \mathbb{F}_2^n \\ x_1, \cdots, x_\ell \in \mathbb{F}_2^n}} |k \rangle_{\text{\uppercase\expandafter{\romannumeral1}}}  \left((-1)^{\langle u_1, x_1 \rangle}  \cdots (-1)^{\langle u_\ell, x_\ell \rangle} \right. & \left.|\eta_1 \cdots \eta_n 0^{(\ell-n)n} \rangle \right)_{\text{\uppercase\expandafter{\romannumeral2}}}\\ 
     &|h(k,x_1,\cdots, x_\ell) \rangle_{\text{\uppercase\expandafter{\romannumeral3}}} 
     |mark \rangle_{\text{\uppercase\expandafter{\romannumeral4}}} |solution) \rangle_{\text{\uppercase\expandafter{\romannumeral5}}},
 \end{aligned}
 \end{equation*}  
where when $dim(span(u_1, \cdots, u_\ell))=n-1$, i.e. $count(mark_j==0)=n-1$, the value of $solution$ is $k_1'$; otherwise, the value of $solution$ will be multiple or only zero. According to Theorem \ref{theorem2}, only when $dim(span(u_1, \cdots, u_\ell))=n-1$, the coefficient matrix $(u_1, \cdots, u_\ell)$ has a unique non-zero solution, so that find $k_1'$ according to $\langle u_i, k_1 \rangle$.

\item Randomly select $\lceil \frac{3m+n\ell}{n}\rceil$ plaintext pairs $(m_i, m_i')$ $(m_i, m_i' \in \mathbb{F}_2^m, m_i \neq m_i')$, prepare the quantum uniform superposition state in the sixth register, and perform inverse transformation $U_4^{-1}$. We can obtain the following result
\begin{equation*}
 \begin{aligned}
      \frac{1}{\sqrt{2^{m+n\ell}\lceil \frac{3m+n\ell}{n} \rceil}} \sum_{\substack{k \in \mathbb{F}_2^m, u_1, \cdots, u_\ell \in \mathbb{F}_2^n \\ x_1, \cdots, x_\ell \in \mathbb{F}_2^n}} |k \rangle_{\text{\uppercase\expandafter{\romannumeral1}}} &\left((-1)^{\langle u_1, x_1 \rangle}  \cdots \right.
     \left.(-1)^{\langle u_\ell, x_\ell \rangle} |u_1 \cdots u_\ell \rangle\right)_{\text{\uppercase\expandafter{\romannumeral2}}} \\ &|h(k,x_1,\cdots, x_\ell) \rangle_{\text{\uppercase\expandafter{\romannumeral3}}}|1 \rangle^{\otimes n}_{\text{\uppercase\expandafter{\romannumeral4}}} 
      |solution \rangle_{\text{\uppercase\expandafter{\romannumeral5}}} \left( \sum_{i=1}^{\lceil \frac{3m+n\ell}{n} \rceil}|m_i\rangle |m_i' \rangle \right)_{\text{\uppercase\expandafter{\romannumeral6}}}. 
 \end{aligned}
 \end{equation*}
At this time, there is entanglement between the second register and the fifth register.
 \item Suppose that there is a quantum oracle $O$, which can calculate the function $f_{k,k_1'}$, $k \in \mathbb{F}_2^m, k_1' \in \mathbb{F}_2^n$,
\[ \left\{ \begin{aligned}
    &f_{k_0,k_1}(k=k_0\land k_1'=k_1)=1, 
    \\  &f_{k_0,k_1}(k\neq k_0 \lor k_1'\neq k_1)=0.
\end{aligned}\right.    \] 
Base on calling the oracle, define the unitary operator $O$:
\[O|k\rangle|k_1'\rangle|y\rangle \equiv |y\oplus f_{k_0,k_1}(k,k_1') \rangle,\]
where $f(k,x)=\text{Enc}(x)\oplus E_k(x)=E_{k_0}(x\oplus k_1)\oplus k_2\oplus E_k(x)$, Enc$(x)$ corresponds to FX-construction Enc$(x)=E_{k_0}(x\oplus k_1)\oplus k_2$, $E_k$ is a random permutation in a secure block cipher \cite{leander2017grover}. When $E_{k_0}(m_i\oplus k_1)\oplus E_{k_0}(m_i'\oplus k_1)=E_k(m_i\oplus k_1')\oplus E_k(m_i'\oplus k_1')$ for all $(m_i,m_i')$, then $f_{k_0,k_1}=1$.
\item  Add the auxiliary qubit $|1\rangle$ and perform Hadamard transformation, then perform Grover iteration $k$ times $Q^k|\psi\rangle$, return the result $k_0, k_1$ by measuring the first and fifth register.
\end{enumerate}

Next, the quantum classifier $\mathcal{S_B}$ in Figure \ref{fig:GQCOGMS} can be constructed, and the detailed quantum circuit of the Figure \ref{fig:GQCOGMS} is as shown in Figure \ref{fig:SQCOGMS}. 
\begin{figure}[H]
    \centering
    \includegraphics[scale=0.4]{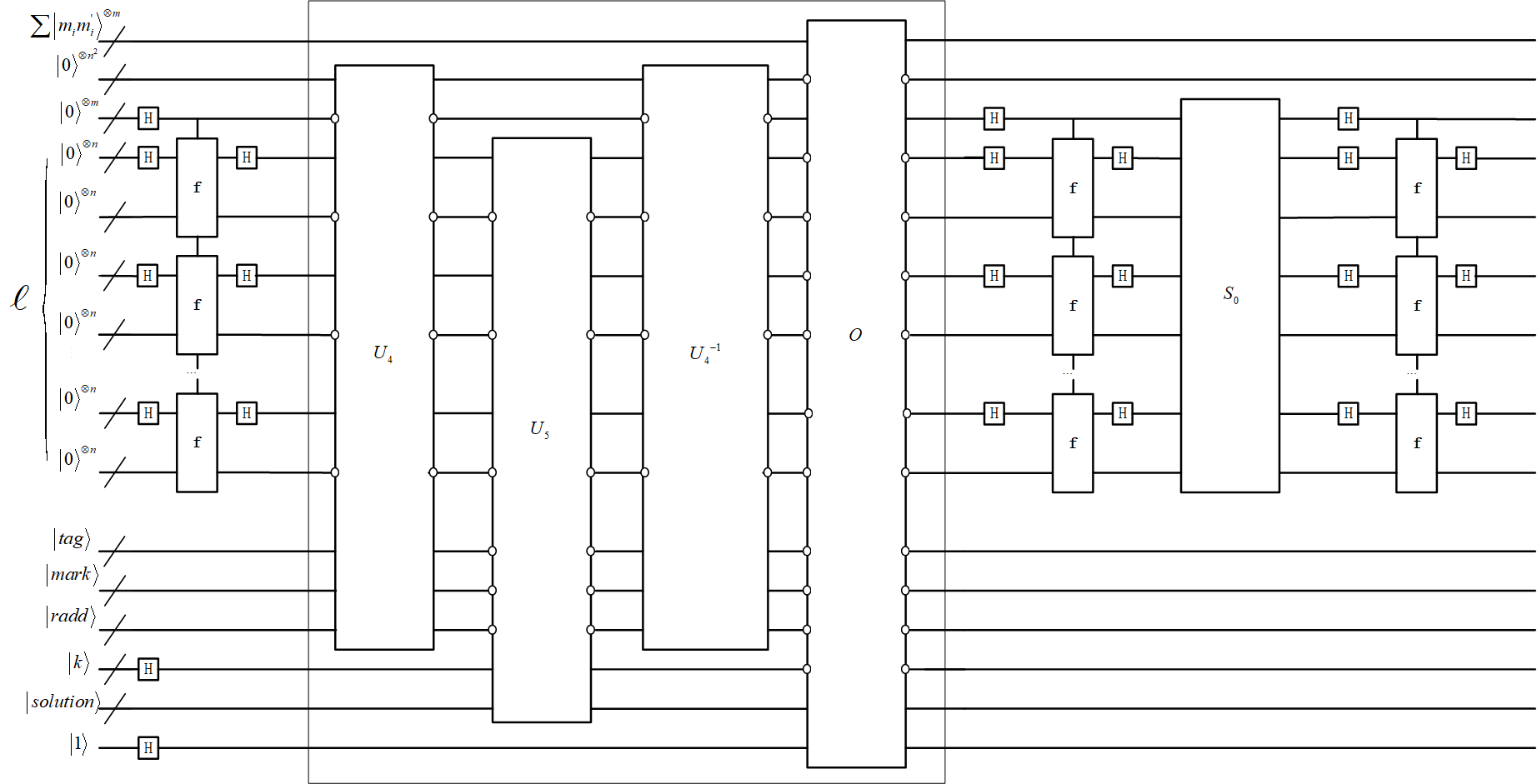}
    \caption{Specific quantum circuit of Grover Meets Simon algorithm}
    \label{fig:SQCOGMS}
\end{figure}

Where, the framed part in Figure \ref{fig:SQCOGMS} is the quantum classifier construction of $\mathcal{S_B}$ in Figure \ref{fig:GQCOGMS}, and $\mathcal{S}_0$ is the unitary operator $2|0 \rangle^{m+2n\ell}\langle 0|^{m+2n\ell}-I_{m+2n\ell}$. The iteration operator is $Q=\mathcal{A}\mathcal{S}_0\mathcal{A}^{-1}\mathcal{S_B}$, where $\mathcal{S_B}=OU_4^{-1}U_5U_4$.

The analysis of the success probability is given in \cite{leander2017grover}, and we summarize it as the following theorem:
\begin{theorem}\label{theorem5}
By choosing $\ell=2(n+\sqrt{n})$ such that the probability that $\langle u_1,\cdots,u_\ell\rangle$  contains at least $n- 1$ proper states is at least $\frac{4}{5}$, and randomly select $\lceil \frac{3m+n\ell}{n}\rceil$ plaintext pairs to make the classifier output 1, the probability of getting $k=k_0$ is at least $1-\frac{1}{2^{2m+n\ell-4}}$. Based on the Quantum Amplitude Amplification Theorem proposed by Brassard, Hoyer, Mosca and Tapp \cite{brassard2002quantum}, choose the number of Grover iterations $k = \lceil\ \frac{\pi}{4\arcsin{2^{-\frac{m}{2 }}}}\rceil$, then the probability of getting a good state $|k_0,k_1\rangle$ is at least $\frac{2}{5}$.
\end{theorem} 

Based on Algorithm \ref{Alg:QAFGSAR}, we can guarantee the same success probability as the above theorem, or even higher (by choosing the number of parallel Simon's algorithms).

\subsection{Application to Asymmetric Search of a Shift}
We introduce a problem that can be viewed as a general combination of Simon's problem and Grover's problem, which can be solved by a corresponding combination of algorithmic ideas and can reduce query complexity of attacking the FX-construction compared with Grover Meets Simon algorithm. First, we introduce a periodic asymmetric search problem described as follows:

\textbf{Asymmetric Search Problem of a Period :} Suppose $F: \{0,1\}^m\times \{0,1\}^n \xrightarrow{}\{0,1\}^\ell$, $g: \{0,1\}^n \xrightarrow{}\{0,1\}^\ell$ are two functions. $F$ is a family of functions on $\{0,1\}^m$, written as $F(i,\cdot)=f_i(\cdot)$, $\exists$ unique $i_0\in \{0,1\}^m$, such that $\forall$ $x\in \{0,1\}^n$, we have $f_{i_0}(x)\oplus g(x)=f_{i_0}(x\oplus s)\oplus g(x\oplus s)$. The goal is to find $i_0$ and $s$.

$\textbf{Simon's promise:}$ Suppose $s$ is a period of the function $f_{i_0}\oplus g$,
$\exists$ $a\in \{0,1\}^n \backslash
 \{0,s\},i\in \{0,1\}^m \backslash
 \{i_0 \}$, such that the input $x\in \{0,1\}^n$ satisfies $(f_i\oplus g)(x\oplus a)=(f_i\oplus g)(x)$. Assume
 \begin{equation}
 \label{equ7}
     \max_{\substack{a\in \{0,1\}^n \backslash
 \{0,s\}  \\  i\in \{0,1\}^m \backslash
 \{i_0 \} }} Pr[(f_i\oplus g)(x\oplus a)=(f_i\oplus g)(x)]  \leq \frac{1}{2}.
 \end{equation}

We analyze Alg-PolyQ2 algorithm in detail, apply Algorithm \ref{Alg:QAFGSARAV} for solving linear equations in section \ref{Algorithms} as a subroutine to Alg-PolyQ2 algorithm \cite{bonnetain2019quantum}, and while ensuring the high probability of solving $i_0$ and period $s$, we construct the circuit of the \textbf{test} oracle and the specific circuit of the entire algorithm.

We show the universal quantum circuit of Alg-PolyQ2 algorithm about the \textbf{test} oracle:
\begin{figure}[H]
    \centering
    \includegraphics[scale=0.55]{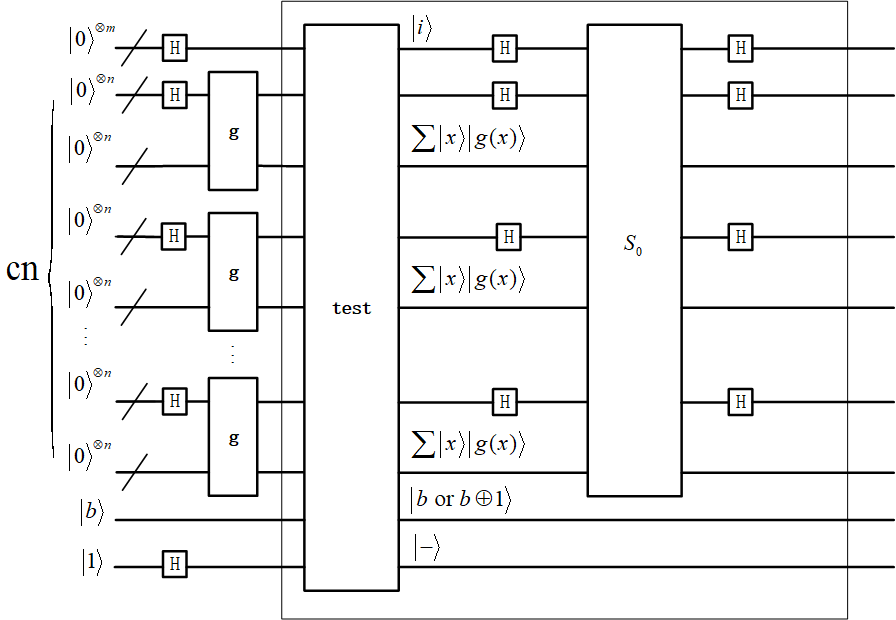}
    \caption{Universal quantum circuit of Alg-PolyQ2 algorithm}
    \label{fig:GQCOTOSA}
\end{figure}
\noindent Where $\mathcal{S}_0$ is unitary operator $2|0\rangle^{\otimes m+2cn^2}\langle 0|^{\otimes m+2cn^2}-I_{m+2cn^2}$; the framed part in Figure \ref{fig:GQCOTOSA} is an iterative operator $(I_{cn^2}H^{\otimes m+cn^2})\mathcal{S}_0 (H^{\otimes m+cn^2}I_{cn^2}) \otimes test$.

Compared with Grover Meets Simon algorithm, this algorithm has two improvements: the first improvement is to reduce the query to $g$, from the exponential level to the polynomial level; the second improvement is that for each $i\in I $, if we get the superposition state $|\psi_g\rangle=\otimes^{cn}\left( \sum_{x\in \{0,1\}^n}|x\rangle|g(x)\rangle \right)$ and $f_i$ can be queried, then we can obtain whether $f_i\oplus g$ has a period without repeatedly querying $g$.

Here, we apply Algorithm \ref{Alg:QAFGSARAV} as a subroutine to this quantum algorithm and give the specific construction of the \textbf{test} oracle. We can obtain $i_0$ and period $s$ by the last measurement. The whole algorithm is as follows:
\begin{enumerate}
    \item Prepare the initial state $|0\rangle^{\otimes 2cn^2}_{\text{\uppercase\expandafter{\romannumeral1}}}|0\rangle^{\otimes m}_{\text{\uppercase\expandafter{\romannumeral2}}}|0\rangle_{\text{\uppercase\expandafter{\romannumeral3}}}$.
    \item Perform the Hadamard transform $H^{\otimes m+cn^2}$ on the first $m+cn^2$ qubits, the result is 
     \[\frac{1}{\sqrt{2^{m+cn^2}}}\otimes^{cn}\left( \sum_{x\in \{0,1\}^n} |x\rangle   \right)_{\text{\uppercase\expandafter{\romannumeral1}}}  \otimes \sum_{i\in \{0,1\}^m} |i\rangle_{\text{\uppercase\expandafter{\romannumeral2}}} \otimes |b\rangle_{\text{\uppercase\expandafter{\romannumeral3}}}.  \]
     where b is initialized to 0.
    \item Query $g$ $cn$ times and obtain the following quantum states:
    \begin{equation*}
        \begin{aligned}
            |\psi_g\rangle_{\text{\uppercase\expandafter{\romannumeral1}}} \otimes \sum_{i\in \{0,1\}^m}|i\rangle_{\text{\uppercase\expandafter{\romannumeral2}}} \otimes |0\rangle_{\text{\uppercase\expandafter{\romannumeral3}}}=\frac{1}{\sqrt{2^{m+cn^2}}}
            \otimes^{cn}\left( \sum_{x\in \{0,1\}^n} |x\rangle |g(x)\rangle   \right)_{\text{\uppercase\expandafter{\romannumeral1}}}  \otimes \sum_{i\in \{0,1\}^m} |i\rangle_{\text{\uppercase\expandafter{\romannumeral2}}} \otimes |0\rangle_{\text{\uppercase\expandafter{\romannumeral3}}}. 
        \end{aligned}
    \end{equation*}
    \item Similarly, query $f$ $cn$ times and obtain the following quantum states:
    \begin{equation*}
        \begin{aligned}
             |\psi_{g\oplus f}\rangle_{\text{\uppercase\expandafter{\romannumeral1}}} \otimes \sum_{i\in \{0,1\}^m}|i\rangle_{\text{\uppercase\expandafter{\romannumeral2}}} \otimes |0\rangle_{\text{\uppercase\expandafter{\romannumeral3}}} =\frac{1}{\sqrt{2^{m+cn^2}}}
            \otimes^{cn}\left( \sum_{x\in \{0,1\}^n} |x\rangle |(g\oplus f)(x)\rangle \right)_{\text{\uppercase\expandafter{\romannumeral1}}} \otimes \sum_{i\in \{0,1\}^m}|i\rangle_{\text{\uppercase\expandafter{\romannumeral2}}} \otimes |0\rangle_{\text{\uppercase\expandafter{\romannumeral3}}}. 
        \end{aligned}
    \end{equation*}
    \item Perform the Hadamard transform on the first $cn^2$ qubits, the result is 
     \begin{equation*}
        \begin{aligned}
           \frac{1}{\sqrt{2^{m+2cn^2}}} \left(\sum_{u_1,x_1\in\{0,1\}^n} (-1)^{u_1\cdot x_1}|u_1\rangle|(g\oplus f)(x_1)\rangle  \otimes \right.    & \left.\cdots \otimes   \sum_{u_{cn},x_{cn}\in\{0,1\}^n} (-1)^{u_{cn}\cdot x_{cn}}|u_{cn}\rangle|(g\oplus f)(x_{cn})\rangle\right) _{\text{\uppercase\expandafter{\romannumeral1}}}   \\
            &\otimes \sum_{i\in \{0,1\}^m}|i\rangle_{\text{\uppercase\expandafter{\romannumeral2}}} \otimes |0\rangle_{\text{\uppercase\expandafter{\romannumeral3}}}.
        \end{aligned}
    \end{equation*}
    \item Since it is necessary to ensure that the register where $|\psi_g\rangle$ is located is disentangled before and after the \textbf{test} oracle, we apply Algorithm \ref{Alg:QAFGSARAV} and use the first register as control qubits to perform $U_4U_5$, store the calculated value of the rank in the fourth register, and store the calculated value of the period in the fifth register. We can obtain the following results (ignore the amplitude, other auxiliary qubits are not listed, the same below):
    \begin{equation*}
        \begin{aligned}
           \sum_{\substack{i \in \mathbb{F}_2^m\\ u_1, \cdots, u_{cn} \in \mathbb{F}_2^n \\ x_1, \cdots, x_{cn} \in \mathbb{F}_2^n}} \left(|0\rangle^{\otimes cn}|(g\oplus f)(x_1)\rangle  \otimes \cdots \otimes   |(g\oplus f)(x_{cn})\rangle\right) _{\text{\uppercase\expandafter{\romannumeral1}}}  
           \otimes |i\rangle_{\text{\uppercase\expandafter{\romannumeral2}}} \otimes |0\rangle_{\text{\uppercase\expandafter{\romannumeral3}}}\otimes |mark\rangle_{\text{\uppercase\expandafter{\romannumeral4}}}\otimes |solution\rangle_{\text{\uppercase\expandafter{\romannumeral5}}}.
        \end{aligned}
    \end{equation*}
    \item Use the fourth register as control qubits (n CNOT gates) and store the value in the third register. The following result is:
   \begin{equation*}
        \begin{aligned}
           \sum_{\substack{i \in \mathbb{F}_2^m\\ u_1, \cdots, u_{cn} \in \mathbb{F}_2^n \\ x_1, \cdots, x_{cn} \in \mathbb{F}_2^n}} \left(|0\rangle^{\otimes cn}|(g\oplus f)(x_1)\rangle  \otimes \cdots \otimes   |(g\oplus f)(x_{cn})\rangle\right) _{\text{\uppercase\expandafter{\romannumeral1}}} 
           \otimes |i\rangle_{\text{\uppercase\expandafter{\romannumeral2}}} \otimes |r\rangle_{\text{\uppercase\expandafter{\romannumeral3}}}\otimes |mark\rangle_{\text{\uppercase\expandafter{\romannumeral4}}}\otimes |solution\rangle_{\text{\uppercase\expandafter{\romannumeral5}}}.
        \end{aligned}
    \end{equation*}
Here, when $dim(span(u_1,\cdots,u_{cn}))=n$, $r=0$; when $dim(span(u_1,\cdots,u_{cn}))<n$, $r=1$.
    \item Perform the inverse transformation $U_4^{-1}$ to obtain the following result:
    \begin{equation*}
        \begin{aligned}
            \sum_{\substack{i \in \mathbb{F}_2^m\\ u_1, \cdots, u_{cn} \in \mathbb{F}_2^n \\ x_1, \cdots, x_{cn} \in \mathbb{F}_2^n}} |u_1\rangle|(g\oplus f)(x_1)\rangle  \otimes \cdots  
            \otimes   |u_{cn}\rangle|(g\oplus f)(x_{cn})\rangle _{\text{\uppercase\expandafter{\romannumeral1}}} 
            \otimes |i\rangle_{\text{\uppercase\expandafter{\romannumeral2}}} \otimes |r\rangle_{\text{\uppercase\expandafter{\romannumeral3}}}\otimes |mark\rangle_{\text{\uppercase\expandafter{\romannumeral4}}}\otimes |solution\rangle_{\text{\uppercase\expandafter{\romannumeral5}}}.
        \end{aligned}
    \end{equation*}
     \item Suppose that there is a quantum oracle $O$, which can calculate the function $f_{i,r}$, $i \in \mathbb{F}_2^m, r \in \mathbb{F}_2$,
\[ \left\{ \begin{aligned}
    &f_{i_0,1}(i=i_0\land r=1)=1, 
    \\  &f_{i_0,1}(i\neq i_0 \lor r\neq 1)=0.
\end{aligned}\right.    \] 
By calling the oracle, one can perform the unitary operator $O$:
\[O|i\rangle|r\rangle|y\rangle \equiv |y\oplus f_{i_0,1}(i,r) \rangle.\]
where $f_i(x)\oplus g(x)=[E_i(x)\oplus E_i(x\oplus 1)] \oplus [\text{Enc}(x)\oplus \text{Enc}(x\oplus 1)]$, Enc$(x)$ corresponds to FX-construction Enc$(x)=E_{i}(x\oplus s)\oplus k_{out}$ \cite{bonnetain2019quantum}. 
    \item  Perform Hadamard transform on the first $cn^2$ qubits to uncompute, and query $f$ again such that $|\psi_{f\oplus g}\rangle$ becomes $|\psi_g\rangle$.
    \item  Add the auxiliary qubit $|1\rangle$ and perform Hadamard transformation, then perform Grover iteration, return the result $i, r$ by measuring the second and third register.
    \item   If the hidden shift $s$ is also required, the instance of parallel Simon's algorithm in section 4.1 needs to be applied to obtain the result $s$ by measuring the fifth register.
\end{enumerate}

Next, the \textbf{test} oracle in Figure \ref{fig:GQCOTOSA} can be constructed, and a detailed quantum circuit based on Grover iterations is shown as in Figure \ref{fig:SQCOTOSA}.
\begin{figure}[H]
    \centering
    \includegraphics[scale=0.43]{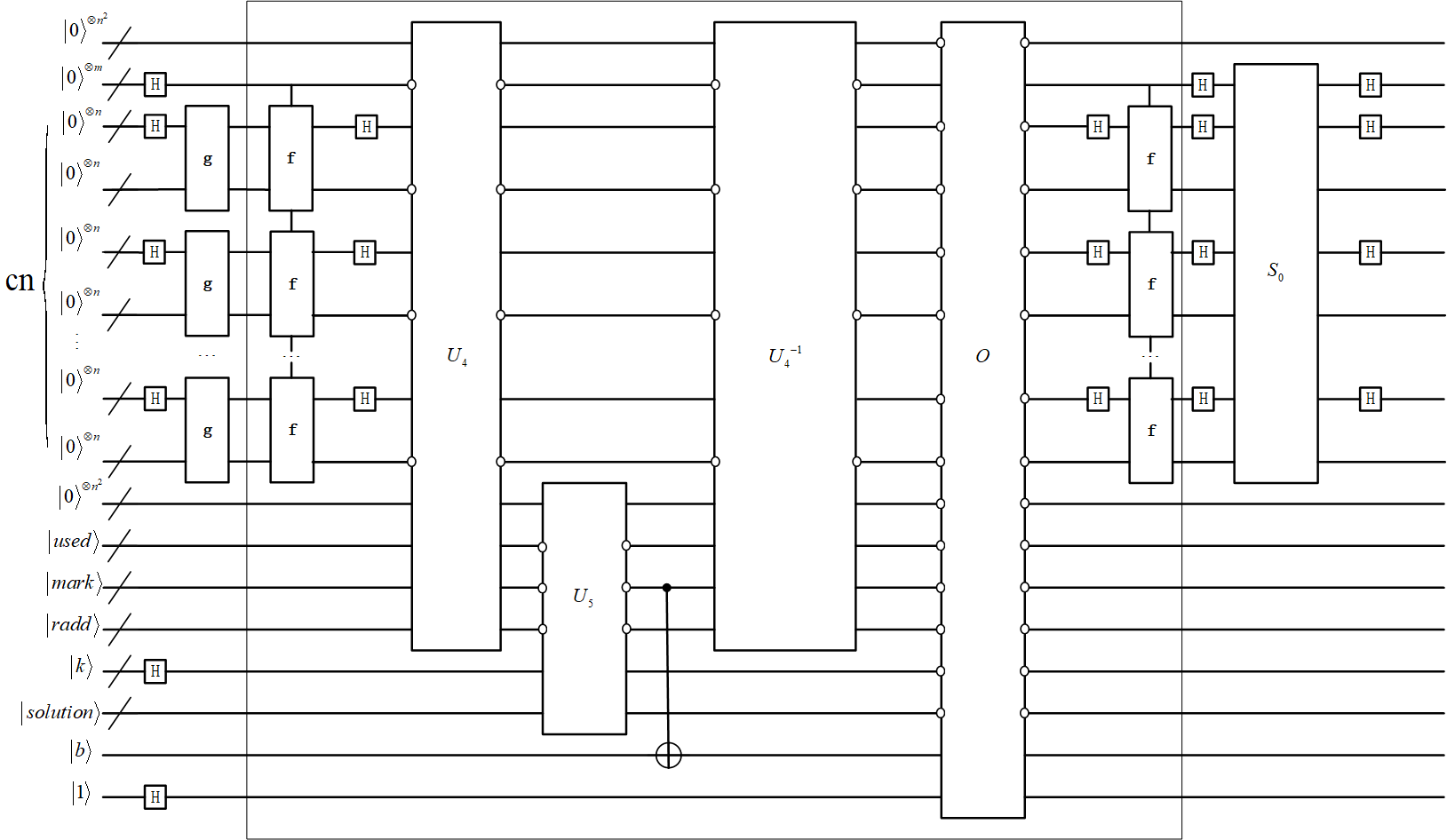}
    \caption{Specific quantum circuit of Alg-PolyQ2 algorithm}
    \label{fig:SQCOTOSA}
\end{figure}
Where we give the specific \textbf{test} oracle construction, $b$ is initialized to 0, and the framed part in Figure \ref{fig:SQCOTOSA} is the \textbf{test} oracle. Before and after applying the \textbf{test}, the register where $|\psi_g\rangle$ should be disentangled.

Similarly, the analysis of the success probability is given in \cite{bonnetain2019quantum}, and we summarize it as the following
theorem:
\begin{theorem}\label{theorem6}
The Alg-PolyQ2 algorithm can query $\mathcal{O}(n)$ times $g$ and $\mathcal{O}(n2^{m/2})$ times $f_i$ to find $i_0$ with probability $\mathcal{ \theta}(1)$, the error produced in each iteration is bounded by the maximum on $i$ of $p^{(i)}=Pr[dim(Span(u_1,\cdots,u_{cn}))<n]$, $i.e.$ the error probability that Simon's algorithm returns $f_i\oplus g$ is a periodic function, but $f_i\oplus g$ is actually not a periodic function is $p^{(i)}\leq 2^ {(n+1)/2}(\frac{3}{4})^{cn/2}$.
\end{theorem} 

Based on Algorithm \ref{Alg:QAFGSARAV}, we can also guarantee the same success probability as the above theorem, or even higher (by choosing the number of parallel Simon's algorithms).

\section{Discussion and Conclusion} \label{Discussion and Conclusion}
\subsection{Discussion}
In this paper, we propose two quantum algorithms \ref{Alg:QAFGSAR} and \ref{Alg:QAFGSARAV} for solving quantum linear equations with coherent superposition. The main difference is that Algorithm \ref{Alg:QAFGSAR} stores the general solution directly in the data register containing the quantum coefficient matrix, and this register is entangled with the solution register. While Algorithm \ref{Alg:QAFGSARAV} has an additional $n(n+1)$ qubits auxiliary register for storing the general solution, and the data register containing quantum coefficient matrix is disentangled before and after input (the data qubits remains unchanged). Therefore, we can select distinct quantum algorithms to apply based on various quantum application scenarios, ensuring the feasibility in different quantum settings.

The number of quantum gates of the proposed quantum algorithms for solving quantum linear systems of equations with coherent superposition is $\mathcal{O}(n^3)$. We guess that the lower bound of the number of quantum gates in solving quantum linear systems of equations with coherent superposition at least $\mathcal{O}(n^3)$, which also shows that our quantum algorithms have reached the optimum. Next,we will give a brief proof.

\begin{proof}(\textbf{Sketch})
We know that solving linear equations in a classical computer is a P problem. Based on the quantum Turing completeness Theorem \ref{theorem3}, we can transform this problem into solving quantum linear equations in polynomial time, and we call it the Q problem. Simon's problem is a BQP problem. Therefore, the problem of solving quantum linear systems of equations can be reduced to Simon's problem. Then, the lower bound of the number of quantum gates for solving the problem of quantum linear systems of equations will not be lower than the quantum gates required by parallel Simon's algorithm.

Similarly, we assume that $f$ is a function that can solve the problem of quantum linear equations. There must be an Oracle machine $M$, such that $M^f$ can solve the parallel Simon's problem. At this time, we cook reduce the parallel Simon's problem to solve the problem of quantum linear equations. Then, the lower bound of the number of quantum gates required by parallel Simon's algorithm will not be lower than the number of quantum gates for solving the problems of quantum linear equations.

To sum up, the lower bound of the number of quantum gates for solving the problem of quantum linear equations should be the same as the number of quantum gates required by parallel Simon's algorithm. However, the number of quantum gates required by parallel $m=\mathcal{O}(n)$ Simon's algorithms is at least $\mathcal{O}(n^2)$, so the lower bound of the number of quantum gates for solving the problem of quantum linear systems of equations is at least $\mathcal{O}(n^2)$.

The main idea of our algorithms is still based on Gaussian-Jordan elimination, which needs to traverse the elements of each row and column and perform transformation operations (this process requires $n$ Toffoli gates). If the lower bound of the required number of quantum gates is $\mathcal{O}(n^2)$, it means that only the row traversal or column traversal is performed, and the entire algorithm process cannot be completed. Therefore, the number of quantum gates of a quantum algorithm for solving a quantum linear system of equations is at least $\mathcal{O}(n^3)$.
\end{proof}

The CNOT gate can be seen as the key to quantum computers, and multiple CNOT gates in an ion trap quantum computer cannot be operated in parallel, but they can only be performed in series. We briefly discuss the applicability of our quantum algorithms based on ion trap computers. An ion trap quantum computer's effective working time (i.e., decoherence time) is no more than $10^3$ seconds, currently 600 seconds \cite{wang2017single}. And it can only handle CNOT gates of the order of $10^2$ in one error-correction period \cite{yang2020effect}. The time of performing a CNOT gate in an ion trap quantum computer is about $2.85 \times 10^{-4}s$ \cite{yang2013post}. Our algorithms give a specific construction for solving quantum linear equations with coherent superposition. If there are $m=\Uptheta(cn)$ quantum linear equations with coherent superposition, the number of CNOT gates reaches $\mathcal{O}((12c+3)n^3)$ after decomposing Toffoli gates into CNOT gates based on section \ref{Algorithms}. Parallel computation of multiple quantum computers still faces many technical and implementation challenges. To complete the attack within a meaningful time, our algorithms should be mainly applied to many attacks against lightweight symmetric cipher constructions (e.g., the Three-round Feistel scheme, Even-Mansour Construction \cite{kuwakado2010quantum}, FX Construction \cite{leander2017grover}, and so on). It can be applied to the detailed construction of the black box that needs to solve the rank of quantum linear equations, such as DESX \cite{kilian1996protect} (a 64-bit state, two 64-bit whitening keys, and a 56-bit internal key), PRINCE \cite{borghoff2012prince}, and PRIDE \cite{albrecht2014block} (one 64-bit state, two 64-bit whitening keys, and one 64-bit internal key).

\subsection{Conclusion}
We first give the definition of quantum linear equations with coherent superposition in detail, then we propose two quantum algorithms for solving quantum linear equations with coherent superposition, and construct their specific quantum circuits. This is a work that has never been done before. According to the quantum Turing completeness Theorem \ref{theorem3}, a quantum computer can theoretically simulate the calculation process of any classical Turing machine. Therefore, we analyze the number of quantum gates simulating the classical Gaussian-Jordan elimination is about $\mathcal{O}(n^4)$. However, the number of quantum gates of applying Algorithm \ref{Alg:QAFGSAR} and Algorithm \ref{Alg:QAFGSARAV} can be reduced to $\mathcal{O}(n^3)$, then we prove that it reaches optimality. Moreover, our quantum algorithms can be applied in different quantum settings. All solutions of the quantum linear systems of equations in the quantum setting can be obtained by measuring the storage register (may be a data register or auxiliary register) once.

We also apply the proposed algorithms to parallel Simon's algorithm \cite{simon1997power} (including with multiple periods), Grover Meets Simon algorithm \cite{leander2017grover} and Alg-PolyQ2 algorithm \cite{bonnetain2019quantum}, under the condition of satisfying the same success probability (even higher by choosing the number of parallel Simon's algorithms), the solutions of the problem can be obtained by one measurement. In addition, based on the proposed quantum algorithms, we reconstruct the classifier black box $\mathcal{B}$ (in Grover Meets Simon algorithm) as a new quantum classifier $\mathcal{S_B}$ and construct the \textbf{test} oracle (the original article does not introduce how to compute the rank in Alg-PolyQ2 algorithm, we can apply Algorithm \ref{Alg:QAFGSARAV} to compute the rank), then construct their quantum circuits in detail, including specific quantum circuits of the three algorithms. Finally, we discuss our algorithms are applicable to lightweight cryptographic attacks during the effective working time of an ion trap quantum computer.

\section*{FUNDING}
This work was supported by Beijing Natural Science Foundation (Grant No.4234084).

\bibliographystyle{unsrt}
\bibliography{ref}

\section*{Appendix} \label{Appendix}
\subsection*{Gaussian-Jordan Elimination over Binary Fields}

We first give a lemma about the relationship between the rank of the matrix and the number of the basic solution vectors, so as to understand the Gaussian-Jordan elimination algorithm over binary fields.
\begin{lemma}\label{lemma3}
\textbf{(Rank of matrix over binary fields)} Given a $m \times n$ matrix $A$ over binary fields, then its rank $r$ equals the number $n$ of unknown variables minus the number $k$ of the basic solution vectors.
\end{lemma}

\begin{proof}
Let $r$ be the rank of $A$, and $k$ be the number of basic solution vectors.

$r$ equals the dimension of the column space of $A$, i.e., the maximum number of linearly independent column vectors of $A$. The matrix $A$ has $n$ columns in total, and the number of free vector columns is $n-r$, so the rank of $A$ is equal to $r$, the number of basic solution vectors (consists of free vector columns) is $n-r$. Therefore, the rank $r$ of the matrix $A$ is equal to the number $n$ of unknown variables  minus the number $k$ of basic solution vectors , i.e., $r=n-k$.
\end{proof}

According to Lemma \ref{lemma3}, we can know that given a $m \times n$ matrix $A$ over binary fields, when $rank(A)=r$, the number of basis vectors of the solution is $n-r$, then its coefficient augmented matrix must be transformed into the following form through Gaussian-Jordan elimination method and row-column transformation:
\begin{align*}
    [A|b]\stackrel{\text{Gaussian-Jordan elimination}}{\longrightarrow} 
\begin{bmatrix} 1 & \dots & 0 &a'_{1,r+1}&\dots& a'_{1,n} & b'_1\\ 
\vdots &\ddots & \vdots & \vdots & \vdots& \vdots& \vdots \\
0 & \dots & 1 & a'_{r,r+1} &\dots & a'_{r,n} & b'_r\\
0 & \dots& 0 & 0 &\dots& 0 & 0 \\
\vdots & \ddots & \vdots & \vdots & \vdots &\vdots & \vdots\\
0 & \dots& 0 & 0 &\dots& 0 & 0 
\end{bmatrix}
,a_{ij}' \in \mathbb{F}_2.
\end{align*}

The $(r+1)$-th to $n$-th columns are the columns where the free variables are located, then the basic solution system of the linear equations is:
\begin{equation*}
    \eta_1=\begin{bmatrix}a'_{1,r+1}\\ \vdots \\ a'_{r,r+1}\\ 1\\0\\ \vdots \\ 0\end{bmatrix},\eta_2=\begin{bmatrix}a'_{1,r+2}\\ \vdots \\ a'_{r,r+2}\\ 0\\1\\ \vdots \\ 0\end{bmatrix},\dots ,\eta_{n-r}=\begin{bmatrix}a'_{1,n}\\ \vdots \\ a'_{r,n}\\ 0\\0\\ \vdots \\ 1\end{bmatrix}
\end{equation*}

The special solution of the linear equations is $x_0=[b'_1,\dots, b'_r,0,\dots,0]^T$. We can obtain the general solution is $x=x_0+k_1\eta_1+\dots+k_{n-r}\eta_{n-r}$, where $k_1,\dots,k_{n-r} \in \mathbb{F}_2$.

Since the above form has undergone a column transformation, the order of the variable vectors corresponding to the special solution and the basic solution system is not sequential. When the general solution is given, the order of the variable vectors still needs to be adjusted. Algorithm \ref{Alg:GEFGS} gives the method for solving the general solution of the linear equations without column transformation, and the order of the variable vectors does not change at this time.

The row echelon matrix is the middle matrix form of the elimination step in the Gaussian elimination method, which is also a critical step. The row reduced form matrix is the final matrix form after the Gauss-Jordan elimination operation. Based on the two forms, after performing elementary transformations on augmented matrices, we give a Gaussian-Jordan elimination algorithm \ref{Alg:GEFGS} for the general solutions of linear equations over binary fields.

\begin{algorithm}[H]
		\caption{\textbf{Gaussian-Jordan elimination for general solution and rank}}
		\label{Alg:GEFGS}
		\begin{algorithmic}[1]
			\Require 
		    $[A|b]$: coefficient augmented matrix belongs to $\mathbb{F}^{m\times(n+1)}_2$, where $A$ is a $m\times n$ matrix 
			\Ensure  the value of a vector $\vec{x}$ such at $A\vec{x}=\vec{b}$
            \State   \textbf{if} $m<n$ \textbf{then}
			\State \ \ \ \  Pivot=list[ ];$x_0$=list[ ];
			\\ \ \ \ \ $\eta$=list[$\eta_1[0^{n-m}     ],\cdots,\eta_{n}[0^{n-m}]$];
			\State   \textbf{else} 
			\State \ \ \ \ Pivot=list[ ];$x_0$=list[ ];$\eta$=list[$\eta_1$[  ],$\cdots,\eta_{n}$[ ]];
			\State   \textbf{for}  $i \leftarrow 1$ $\textbf{to}$ $m$ \textbf{do}
			\State \ \ \ \   \textbf{for} $j\leftarrow 1 $     $\textbf{to}$ $n$ \textbf{do}
			\State \ \ \ \ \ \ \ \      \textbf{if} $\forall_{1\leq t<j}$ $a_{ik}==0$ \textbf{then}
			\State \ \ \ \ \ \ \ \ \ \ \ \   Pivot.append(j)
			\State \ \ \ \ \ \ \ \ \ \ \ \   \textbf{for} $\ell \leftarrow 1$  $\textbf{to}$ $i-1$ and $\ell \leftarrow i+1$  $\textbf{to}$ $m$
			\State \ \ \ \ \ \ \ \ \ \ \ \ \ \ \ \  \textbf{if} $a_{\ell j}=1$ \textbf{then}
		    \State \ \ \ \ \ \ \ \ \ \ \ \ \ \ \ \ \ \ \ \    $a_{\ell:}=a_{i:}\oplus a_{\ell:}$; $b_{\ell}=b_{i}\oplus b_{\ell}$;
		    \State \ \ \ \ \ \ \ \ \ \ \ \  $swap(a_{i:};a_{j:})$; $swap(b_{i};b_{j})$;
		    \State  \textbf{for}  $j \leftarrow 1$ $\textbf{to}$ $n$  \textbf{do}
			\State \ \ \ \      \textbf{if} $j$ in Pivot \textbf{then}
			\State \ \ \ \ \ \ \ \   $x_0.append(b_j)$;
			\State \ \ \ \ \ \ \ \   $\eta_j.append(0^{\min(m,n)})$;
			\State \ \ \ \       \textbf{else} 
			\State \ \ \ \ \ \ \ \       $x_0.append(0)$
			\State \ \ \ \ \ \ \ \       \textbf{for}  $i \leftarrow min(m,n)$ $\textbf{to}$ $1$  \textbf{do}
			\State \ \ \ \ \ \ \ \ \ \ \ \    $\eta_j$.insert(0,$a_{ij}$);
			\State \ \ \ \ \ \ \ \ \ \ \ \      $\eta_j[j]=1$;
			\State  \textbf{return} $x=x_0+k_1\eta_1+\cdots+k_{n}\eta_{n}$; $rank(A)=len(\text{Pivot})$
		\end{algorithmic}
\end{algorithm}

Brief description of Algorithm \ref{Alg:GEFGS}: Initialization the list, including Pivot, $x_0$, and base solution system $\eta$ list. If $m<n$, use the all-zero elements to complement the list $\eta$ into a $n\times n$ matrix; if $m \geq n$, the matrix remains unchanged. First traverse by row and then traverse by column, if $a_{ij}$ is the first element with 1 in the $i$-th row, mark it as the pivot and store the column index $j$ in the list Pivot. Use the $i$-th row to eliminate the row where the element 1 is located in the $j$-th column. Swap the $j$-th row and the $i$-th row so that all pivots are on the main diagonal. After traversing all the rows and columns, at this time, Pivot stores the index $j$ of the columns where all the pivots are located in row order. Afterwards, continue to traverse by column, if $j$ is in the list Pivot, it means that there is a pivot element in this column, then traverse by row, find the pivot, swap the $i$-th row and the $j$-th row such that the pivot is on the main diagonal, and store the special solution $b_j$ corresponding to $x_j$ in the list $x_0$ at this time, otherwise add zero to the list $x_0$. If $j$ is not in the list Pivot, add the elements in this column to the list $\eta_j$ in reverse order, and assign the $j$-th element in $\eta_j$ to 1. Finally, we can obtain the general solution $x=x_0+k_1\eta_1+\cdots+k_{n}\eta_{n}$, and the rank is the length of the list Pivot.

\subsection*{A Proof of Theorem \ref{theorem4}}
\begin{proof}
Given a lemma in \cite{kaplan2016breaking}, as follows:
\begin{lemma}\label{lemma4}
For $a\in \{0,1 \}^n$, consider function $g(x)=2^{-n}\sum_{y\in a^{\perp}}(-1)^{x\cdot y}$, Where, $a^{\perp}=\{y\in \{0,1\}^n, s.t. y\cdot a=0\}$. For $\forall$ $x$, satisfy
\begin{equation}
\label{equ8}
    g(x)=\frac{1}{2}(\delta_{x,0}+\delta_{x,a})
\end{equation}
Where, $\delta_{x,t}$ means 1 when $x=t$, otherwise 0.
\end{lemma}

Now we start to prove Theorem \ref{theorem4}. If there are $m$ Simon's algorithms in parallel, $m$ vectors $y_1,\cdots,y_m$ can be obtained. Under the promise that the Simon's algorithm is satisfied, these $m$ vectors and $s$ are orthogonally spanned into a $n-1$ dimensional linear space. However, if the spanned space is less than n-1 dimensional, then $s$ can still be efficiently recovered by testing all vectors that are orthogonal to the subspace. Therefore, the probability of not recovering $s$ correctly is $Pr[dim(Span(y_1,\cdots,y_m))\leq n-2]$, as follows:
\begin{equation*}
    \begin{aligned}
        &Pr[y\cdot a=0]\\
        &=\parallel 2^{-n} \sum_{\substack{y\in\{0,1\}^n \\ s.t. y\cdot t=0}}|y\rangle \sum_{x\in\{0,1\}^n}(-1)^{x\cdot y}|f(x)\rangle\parallel^2 \\
        &=2^{-2n}\sum_{\substack{y\in\{0,1\}^n \\ s.t. y\cdot t=0}}\sum_{x,x'\in\{0,1\}^n}(-1)^{(x\oplus x')\cdot y}\langle f(x')|f(x)\rangle\\
        &\overset{\text{eq.(\ref{equ8})}}{=}2^{-2n}\sum_{x,x'\in\{0,1\}^n}\langle f(x')|f(x)\rangle 2^{n-1}(\delta_{x,x'}+\delta_{x,x'\oplus a})  \\
        &=2^{-(n+1)}\left[\sum_{x\in\{0,1\}^n}\langle f(x)|f(x)\rangle + \sum_{x\in\{0,1\}^n}\langle f(x\oplus a)|f(x)\rangle         \right]\\
        &=\frac{1}{2}[1+Pr[f(x)=f(x\oplus a)]]\\
        &\leq \frac{1}{2}(1+\varepsilon(f))\\
        &\leq \frac{1}{2}(1+p_0)
    \end{aligned}
\end{equation*}

Thus, the probability that $s$ cannot be recovered correctly, that is, the failure probability is
\begin{equation*}
    \begin{aligned}
        &Pr[dim(Span(y_1,\cdots,y_m))\leq n-2]\\
        &\leq Pr[\exists a \in \{0,1\}^n \backslash \{0,s\} \ s.t. y_1\cdot a=\cdots=y_m\cdot a=0]\\
        &\leq \sum_{a \in \{0,1\}^n \backslash \{0,s\}}Pr[y_1\cdot a=\cdots=y_m\cdot a=0]\\
        &\leq \sum_{a \in \{0,1\}^n \backslash \{0,s\}}(Pr[y_1\cdot a=0])^m\\
        &\leq \max_{a \in \{0,1\}^n \backslash \{0,s\}}2^n (Pr[y_1\cdot a=0])^m\\
        &\leq 2^n(\frac{1+\varepsilon(f)}{2})^m\\
        &\leq 2^n(\frac{1+p_0}{2})^m
    \end{aligned}
\end{equation*}

Therefore, the success probability of recovering the correct period $s$ by $m$ Simon's algorithms in parallel is at least $1-2^n(\frac{1+p_0}{2})^m$.
\end{proof}

\end{document}